\documentclass[%
 jmp,
cp,  
 amsmath,amsthm,amssymb,
 reprint,%
onecolumn]{revtex4-2}

\usepackage{graphicx}
\usepackage{dcolumn}
\usepackage{bm}

\usepackage[utf8]{inputenc}
\usepackage[T1]{fontenc}
\usepackage{mathptmx} 

\usepackage{amsthm}

\usepackage{hyperref}
\hypersetup{
    citecolor  = blue,
        colorlinks = true,
    }

\usepackage[dvipsnames]{xcolor}

\newtheorem{theorem}{Theorem}[section]
\newtheorem{corollary}[theorem]{Corollary}
\newtheorem{lemma}[theorem]{Lemma}
\newtheorem{proposition}[theorem]{Proposition}

\renewcommand{\Re}{\operatorname{Re}}
\renewcommand{\Im}{\operatorname{Im}}

\usepackage{tikz}
\usepackage{mathrsfs}  

\begin{document}

\title{Edge behavior of higher complex-dimensional determinantal point processes}


\author{L. D. Molag}
 \email{l.d.molag@sussex.ac.uk}
\affiliation{Department of Mathematics, University of Sussex, Brighton, BN1 9RH, United Kingdom 
%
}%

\date{\today} 

\begin{abstract}
As recently proved in generality by Hedenmalm and Wennman, it is a universal behavior of complex random normal matrix models that one finds a complementary error function behavior at the boundary (also called edge) of the droplet as the matrix size increases. Such behavior is seen both in the density of the eigenvalues, and the correlation kernel, where the Faddeeva plasma kernel emerges. These results are neatly expressed with the help of the outward unit normal vector on the edge. We prove that such universal behaviors transcend this class of random normal matrices, being also valid in a specific ``elliptic'' class of determinantal point processes defined on $\mathbb C^d$, which are higher dimensional generalizations of the determinantal point processes describing the eigenvalues of the complex Ginibre ensemble and the complex elliptic Ginibre ensemble. These models describe a system of particles in $\mathbb C^d$ with mutual repulsion, that are confined to the origin by an external field $\mathscr V(z) = |z|^2 - \tau \Re(z_1^2+\ldots+z_d^2)$, where $0\leq \tau<1$. 
Their average density of particles converges to a uniform law on a $2d$-dimensional ellipsoidal region. It is on the hyperellipsoid bounding this region that we find a complementary error function behavior and the Faddeeva plasma kernel. To the best of our knowledge, this is the first instance of the Faddeeva plasma kernel emerging in a higher dimensional model. The results provide evidence for a possible edge universality theorem for determinantal point processes on $\mathbb C^d$. 
\end{abstract}

\maketitle


\section{Introduction} \label{sec:1}

In this paper, we consider a specific ``elliptic'' subclass of determinantal point processes (DPP) on $\mathbb C^d$, with correlation kernel given by
\begin{align} \label{eq:defKnMoreGeneralB}
\mathscr K_n(z, w) &= \sqrt{\mathscr W(z) \mathscr W(w)} \sum_{j\in J_n} \mathscr P_j(z) \overline{\mathscr P_j(w)}, & z, w\in\mathbb C^d,
\end{align}
where $\mathscr W:\mathbb C^d\to [0,\infty)$  is a weight function, $\mathscr P_j(z) = \mathscr P_j(z_1, \ldots, z_d)$ are multivariate orthogonal polynomials, satisfying 
\begin{align*}
\int_{\mathbb C^d} \mathscr P_i(z) \overline{\mathscr P_j(z)} \mathscr W(z) \, d^{2d}z &= \delta_{ij}, & i,j\in J_n,
\end{align*}
and they are indexed by some set $J_n$ such that $\{\mathscr P_j(z_1, \ldots, z_d) : j\in J_n\}$ forms a basis of the space of all polynomials with complex coefficients of total degree smaller than $n$. Of course, it is assumed here that $\mathscr W(z)$ vanishes sufficiently fast as $|z|\to\infty$, such that the above inner products are well-defined. The number of points of the DPP is given by
\begin{align*}
\int_{\mathbb C^d} \mathscr K_n(z,z) d^{2d}z = |J_n| = \binom{n+d-1}{d},
\end{align*}
where the last step follows by a straightforward combinatorial argument.
The positivity requirement for the expressions
\begin{align*}
\det \left(\mathscr K_n(z_i, z_j)\right)_{1\leq i,j\leq k}, \qquad k=1,2,\ldots
\end{align*}
follows from the Cauchy-Binet formula for the determinant of a product of rectangular matrices. 
The subclass that we shall investigate, is called elliptic due to its relation to the complex elliptic Ginibre ensemble \cite{Girko}, and corresponds to the choice
\begin{align} \label{eq:startDefWeight}
\mathscr W(z) = e^{-\mathscr V(z)}, \qquad \mathscr V(z) = |z|^2 - \tau\Re\left(z_1^2+\ldots+z_d^2\right), \qquad 0\leq \tau<1.
\end{align}
Kernels in higher dimensions of the type \eqref{eq:defKnMoreGeneralB} have been studied by Berman \cite{Berman}, who considered the case of exponentially varying weights $\mathscr W(z)=e^{-n \mathscr V(z)}$, for a large class of external fields $\mathscr V:\mathbb C^d\to \mathbb R$ satisfying certain growth and regularity conditions. The kernel in \eqref{eq:defKnMoreGeneralB}, excluding the factor $\sqrt{\mathscr W(z) \mathscr W(w)}$, is called the \textit{Bergman kernel} of the Hilbert space of all polynomials $\mathscr P(z_1, \ldots, z_d)$ of total degree $<n$ with weigthed norm
\begin{align*}
\lVert \mathscr P \rVert^2 = \int_{\mathbb C^d} |\mathscr P(z_1, \ldots, z_d)|^2 \mathscr W(z) \, d^{2d}z.
\end{align*}
Under the growth and regularity conditions assumed in \cite{Berman}, Berman was able to determine that the average density of points converges to the so-called Monge-Ampère measure (see \cite{Kolodziej} for its definition in the language of complex manifolds), which has a compact support
depending on $\mathscr V$. For $d=1$, this compact support is called the \textit{droplet}. We shall adopt this terminology for the $d>1$ case. In the current paper, we will be interested in the behavior of the correlation kernel on the boundary, or \textit{edge}, of the droplet. The particular choice of weight \eqref{eq:startDefWeight} corresponds to the special case where the weight factorizes in identical planar weights, 
\begin{align} \label{eq:Wfactorization}
\mathscr W(z) = \prod_{k=1}^d \omega(z_k),
\end{align}
where $\omega:\mathbb C\to [0,\infty)$ is a planar weight, which in our specific case is given by $\omega(z) = e^{-|z|^2+\tau \Re(z^2)}$. 
Note that, for general factorized weight as in \eqref{eq:Wfactorization}, the multivariate orthogonal polynomials are given simply by
\begin{align*}
\mathscr P_j(z) &= \prod_{k=1}^d P_{j_k}(z_k), & j\in \{0,1,\ldots\}^d, \quad |j|=j_1+\ldots+j_d<n,
\end{align*}
where the $P_j$ have degree $j$ and satisfy the planar orthogonality conditions
\begin{align*}
\int_{\mathbb C} P_i(z) \overline{P_j(z)} \omega(z) \, d^2z &= \delta_{ij}, & i,j=0,1,\ldots, n-1. 
\end{align*}
To motivate our choice of planar weights $\omega$, let us first consider the case of $d=1$ in \eqref{eq:defKnMoreGeneralB}. In this case the kernel is a sum over a product of two planar orthogonal polynomials. It is well-known that the corresponding DPP describes the eigenvalues of so-called \textit{random normal matrices} for a general class of exponentially varying weights $\omega$. These models consist of $n\times n$ complex normal matrices, distributed by
\begin{align*}
\frac{1}{Z_n^V} e^{- n \operatorname{Tr} V(M)} dM, \qquad dM = \prod_{i,j=1}^n d\operatorname{Re} M_{ij} \, d\operatorname{Im} M_{ij},
\end{align*}
where $Z_n^V$ is the normalization constant, and $V$ is a function, the \textit{external field}, satisfying certain growth and regularity conditions. The corresponding eigenvalues are then distributed by
\begin{align*}
\frac{1}{c_n^V} \prod_{1\leq i<j\leq n} |z_i-z_j|^2 \prod_{j=1}^n \omega(z_j) d^2z_j, \qquad z_1,\ldots, z_n\in\mathbb C, 
\end{align*}
where $c_n^V$ is the normalization constant, and we have the planar weight $\omega(z) = e^{-n V(z)}$. The corresponding eigenvalues form a DPP, described by a correlation kernel of the form \eqref{eq:defKnMoreGeneralB} for $d=1$ with $\mathscr W(z)=\omega(z)=e^{-n V(z)}$. 
One major difference with, e.g., the Hermitian ensembles with general external fields and unitary symmetry, is that the eigenvalues of random normal matrices are not necessarily real. Since, with probability $1$, the matrices sampled from the ensemble are not Hermitian, they are called non-Hermitian random matrices. We shall focus on two such random normal matrix models, the (complex) Ginibre ensemble and the elliptic Ginibre ensemble. The normality condition is not essential for these models, and, since they are defined by homogeneous external fields $V$, the $n$ in the exponent of the weights $\omega(z)=e^{-n V(z)}$ is also not essential. The Ginibre ensemble is then defined as the $n\times n$ complex matrices $M$ with elements distributed as i.i.d. complex Gaussian random variables of mean zero and unit variance. Explicitly, these $n\times n$ complex matrices $M$ are distributed by
\begin{align*}
\frac{1}{Z_n} e^{-\operatorname{Tr}(M^* M)} dM. 
\end{align*} 
It was introduced by Ginibre in 1965 \cite{Ginibre}, along with a real and quaternion version. It is easy to check that we have $\mathscr W(z) = \omega(z) = e^{-|z|^2}$ (the external field is thus $V(z)=|z|^2$) and $\mathscr P_j(z) = P_j(z) = \frac{1}{\sqrt{j! \pi}} z^{j-1}$ in the notation of \eqref{eq:defKnMoreGeneralB}, and the eigenvalues of the Ginibre ensemble thus form a DPP on $\mathbb C$, with correlation kernel
\begin{align} \label{eq:defGinibreKer}
\mathscr K_n(z, w) = \frac{1}{\pi} e^{-\frac{|z|^2+|w|^2}{2}} \sum_{j=0}^{n-1} \frac{(z\overline{w})^j}{j!}.
\end{align}

Ginibre proved that the average density of particles $\rho_n^{(1)}(z) = \frac{1}{n} \mathscr K_n(z,z)$ converges to a uniform distribution on the unit disc under a proper scaling \cite{Ginibre}. Namely, we have
\begin{align} \label{eq:densityGinibreScaling}
\lim_{n\to\infty} n \rho_n^{(1)}(\sqrt n \, z)
= \begin{cases} \frac{1}{\pi}, & |z|<1,\\ \frac{1}{2\pi}, & |z|=1,\\ 0, & |z|>1.\end{cases} 
\end{align}
(The case $|z|=1$ was not treated in \cite{Ginibre} though.) The region $|z|\leq 1$ is the droplet, the interior $|z|<1$ is often called the \textit{bulk}. Indeed, the boundary $|z|=1$ is the edge. The pointwise limit for $|z|=1$ in \eqref{eq:densityGinibreScaling} is actually a direct consequence of a more general result concerning edge scaling limits. For any $z$ in the unit circle, and any $u, v\in\mathbb C$, we have
\begin{align} \label{eq:edgeScalingErfcGinibre}
\lim_{n\to\infty} \mathscr K_n\left(\sqrt n \, z + u, \sqrt n \, z + v\right)
= \frac{1}{2\pi} e^{u \overline{v}-\frac{|u|^2+|v|^2}{2}} \operatorname{erfc}\left(\frac{\overline{v} z + u \overline{z}}{\sqrt 2}\right),
\end{align}
where $\operatorname{erfc}$ is the complementary error function, given by
\begin{align*}
\operatorname{erfc}(z) = 1 - \frac{2}{\sqrt \pi} \int_0^z e^{-t^2} dt. 
\end{align*}
Alternatively, we may write \eqref{eq:edgeScalingErfcGinibre} as
\begin{align} \label{eq:edgeScalingErfcGinibreRewrite}
\lim_{n\to\infty} \mathscr K_n\left((\sqrt{n} + u) z, (\sqrt{n} + v) z\right)
= \frac{1}{2\pi} e^{u \overline{v}-\frac{|u|^2+|v|^2}{2}} \operatorname{erfc}\left(\frac{u+\overline v}{\sqrt 2 }\right).
\end{align}
The limiting kernel on the RHS in \eqref{eq:edgeScalingErfcGinibreRewrite} is known as the Faddeeva plasma kernel \cite{HeWe} (note that Hedenmalm and Wenmann use a different definition for the error function), and seems to have been first derived as a scaling limit in \cite{FoHo}. As explained in \cite{HeWe}, the name derives from the plasma dispersion function, which was first tabulated by Faddeeva and Terent\'{ }ev \cite{Faddeeva}.
This edge behavior is known to be universal, i.e., it arises as a scaling limit in a large class of other models. For example, Tao and Vu proved that it holds in the vague topology for random normal matrix models (real or complex), that match up to four moments with the Ginibre ensemble \cite{TV}. It was recently proved by Cipolloni, Erd\H{o}s and Schröder that this four moment condition can be removed when the moments are finite and some growth condition is satisfied by the probability density of the matrix entries \cite{CiErSc}. Ameur, Kang, and Makarov proved that it holds for random normal matrix models with radially symmetric external field \cite{AmKaMa}. 
Very recently, it was proved by Hedenmalm and Wennman that the behavior is universal in random normal matrix models, under weak assumptions on the external field \cite{HeWe} (being ``1-admissible''), essentially bringing the questions about edge universality of these models to their logical conclusion. 

Let us now turn our attention to the elliptic Ginibre ensemble of $n\times n$ complex matrices with parameter $\tau$ \cite{Girko}. These are distributed according to
\begin{align*}
\frac{1}{Z_n} e^{-\frac{1}{1-\tau^2}\operatorname{Tr}\left(M^*M - \frac{\tau}{2}(M^2+(M^*)^2)\right)}.
\end{align*}
The model describes random matrices of the form $M = \sqrt{1+\tau} G_1 + i \sqrt{1-\tau} G_2$, where $G_1, G_2$ are $n\times n$ complex Hermitian matrices picked from the GUE. As such, the elliptic Ginibre ensemble interpolates between the Ginibre ensemble ($\tau=0$) and the GUE ($\tau=1$). For esthetic reasons, we apply a scaling $M \to (1-\tau^2)^{-\frac{1}{2}} M$, and the model is then described by the external field $V(z) = |z|^2-\tau \operatorname{Re}(z^2)$. Consequently, we have $\mathscr W(z)=\omega(z)=e^{-|z|^2+\tau \operatorname{Re}(z^2)}$, and then
\begin{align*}
\mathscr P_j(z) = P_j(z) =
\frac{\sqrt{1-\tau^2}}{\pi} \frac{\left(\frac{\tau}{2}\right)^{\frac{j}{2}}}{\sqrt{j!}} H_{j}\left(\sqrt{\frac{1-\tau^2}{2\tau}} z\right),
\end{align*}
where $H_j(z) = (-1)^j e^{z^2} \frac{d^j}{dz^j} e^{-z^2}$ is the Hermite polynomial of degree $j$. The corresponding eigenvalues form a DPP, where now the correlation kernel is given by
\begin{align} \label{eq:defKerneleGinUE}
\mathscr K_n(z, w) = \frac{\sqrt{1-\tau^2}}{\pi} \sqrt{\omega(z) \omega(w)} 
\sum_{j=0}^{n-1} \frac{1}{j!} \left(\frac{\tau}{2}\right)^j H_j\left(\sqrt{\frac{1-\tau^2}{2\tau}} z\right) \overline{H_j\left(\sqrt{\frac{1-\tau^2}{2\tau}} w\right)}.
\end{align}
 The average density of particles also converges to a uniform distribution, where now the bulk is not given by the unit disc, but by the elliptic domain
\begin{align*}
\mathscr E_\tau = \left\{z\in\mathbb C : \frac{1-\tau}{1+\tau} \Re(z)^2 + \frac{1+\tau}{1-\tau} \Im(z)^2<1 \right\}. 
\end{align*}
Lee and Riser were able to obtain fine asymptotics for the density $\rho_n^{(1)}$ of the Ginibre ensemble and the elliptic Ginibre ensemble \cite{LeRi}. Picking $z\in \partial \mathscr E_\tau$, and letting $\textbf{n}$ denote the outward normal vector on $\partial \mathscr E_\tau$ at $z$, they showed that
\begin{multline} \label{eq:densityLeeRiser}
n \rho_n^{(1)}(\sqrt n \, z + \lambda \textbf{n}) = \frac{1}{2\pi} \operatorname{erfc}(\sqrt 2 \lambda)
+ \frac{\kappa}{\sqrt n} \frac{1}{3\sqrt{2\pi^3}} (\lambda^2-1) e^{-2\lambda^2}\\
+ \frac{1}{n} \frac{e^{-2\lambda^2}}{\sqrt{2\pi^3}} \left(\kappa^2\frac{2\lambda^5-8\lambda^3+3\lambda}{18}
+\left(\frac{(\partial_s\kappa)^2}{9\kappa^2}-\frac{\partial^2_s\kappa}{12\kappa}\right)\lambda\right) + \mathcal O(n^{-\frac{3}{2}+9\nu}),
\end{multline}
where $\kappa$ is the curvature of the ellipse $\partial \mathscr E_\tau$ in $z$, and we have any fixed $0<\nu<\frac{1}{6}$ ($\partial_s$ denotes the derivative with respect to the arclength parameter, see \cite{LeRi} for details). The error bound is uniform for $\lambda\in\mathbb R$ with $\lambda = \mathcal O(n^\nu)$ and $z\in\partial \mathscr E_\tau$. The first treatment of the off-diagonal case for the elliptic Ginibre ensemble can be found in Riser's thesis \cite{Ri}, and its error bounds were recently refined by Byun and Ebke in a paper about the quaternion elliptic Ginibre ensemble \cite{ByEb}. Namely, for some explicit unimodular factors $c_n : \partial \mathscr E_\tau \times \mathbb C\to \mathbb T$ (where $\mathbb T=\{z\in\mathbb C : |z|=1\}$), and any $\varepsilon>0$, we have
\begin{multline} \label{eq:FaddeevaRiserByunEbke}
c_n(z, u) \overline{c_n(z, v)} \mathscr K_n(\sqrt n \, z+u \, \textbf{n}, \sqrt n \, z+v \, \textbf{n})
= \frac{1}{2 \pi} \exp\left(u \overline v-\frac{|u|^2+|v|^2}{2}\right) \operatorname{erfc}\left(\frac{u + \overline v}{\sqrt 2}\right) \\
 +  \frac{1}{\sqrt n} \exp\left(- \frac{|u|^2+u^2+|v|^2+\overline{v}^2}{2}\right) 
 \kappa \frac{u^2+\overline v^2 - u\overline v - 1}{3\sqrt{2 \pi^3}} 
   + \mathcal O(n^{-1+\varepsilon}),
\end{multline}
as $n\to\infty$, uniformly for $z\in \partial\mathscr E_\tau$, and $u, v\in\mathbb C$ bounded (in fact, when $0<\varepsilon<\frac{1}{2}$, the condition $u, v = \mathcal O(n^{\varepsilon/3})$ is sufficient for the error bound to be uniform). Here $\kappa$ is the curvature of $\partial \mathscr E_\tau$ at $z$. Note that these unimodular factors drop out when calculating determinants, and are thus irrelevant for the $k$-point correlation functions.


Returning now to the general case $d\geq 1$, our goal is to show that
edge scaling limits such as \eqref{eq:densityLeeRiser} and \eqref{eq:FaddeevaRiserByunEbke} also hold for higher dimensional DPPs with correlation kernel of the form \eqref{eq:defKnMoreGeneralB}, where $\mathscr W$ factorizes as in \eqref{eq:Wfactorization} with the planar weights $\omega(z) = e^{-|z|^2+\tau \operatorname{Re}(z^2)}$ from either the Ginibre ($\tau=0$) or the elliptic Ginibre ensemble ($0<\tau<1$). Explicitly, for $\tau=0$, this means that we consider 
\begin{align} \label{eq:defKernelGinUEd}
\mathscr K_n(z, w) &= \frac{1}{\pi^d} e^{-\frac{|z|^2+|w|^2}{2}} \sum_{|j|<n} \prod_{k=1}^d \frac{(z_k \overline w_k)^{j_k}}{j_k!},
& z, w\in\mathbb C^d,
\end{align}
while for $0<\tau<1$, we take
\begin{align} \label{eq:defKerneleGinUEd}
\mathscr K_n(z, w) = \left(\frac{\sqrt{1-\tau^2}}{\pi}\right)^d  
\sum_{|j|< n} \frac{\left(\frac{\tau}{2}\right)^{|j|}}{j_1! \cdots j_d!} \prod_{k=1}^d \sqrt{\omega(z_k) \omega(w_k)} H_{j_k}\left(\sqrt{\frac{1-\tau^2}{2\tau}} z_k\right) \overline{H_{j_k}\left(\sqrt{\frac{1-\tau^2}{2\tau}} w_k\right)}, \qquad z, w\in \mathbb C^d.
\end{align}
Here the summation is over multi-indices $j=(j_1, \ldots, j_d)\in \{0,1,\ldots\}^d$ such that $|j|=j_1+\ldots+j_d<n$. The models describe a system of particles in $\mathbb C^d$ that repel each other (due to the determinantal structure), which are confined to the origin due to the external field $\mathscr V(z) = |z|^2 - \tau \Re(z_1^2+\ldots+z_d^2)$.
One motivation for studying such higher dimensional processes comes from \cite{ADM}. There, the DPP on $\mathbb C^d$ with correlation kernel \eqref{eq:defKerneleGinUEd}, was considered.  
In the limit $\tau\to 1$, this model corresponds to spinless free fermions in $\mathbb R^d$ in a harmonic potential \cite{DeDoMaSc3}. 
One major reason to investigate the model for $d>1$, was that it allowed to probe a weak non-Hermiticity regime \cite{FyKhSo1, FyKhSo2, FyKhSo3}. For $d=1$ it is known that interpolating kernels are found in the bulk and on the edge in the weak non-Hermiticity regime \cite{Bender, ABe, AmBy}. Recently, an interpolating kernel was also found for the rightmost eigenvalue of the elliptic Ginibre ensemble in the weak non-Hermiticity regime \cite{BoLi}. Higher dimensional versions of such interpolating kernels for the bulk and edge were indeed found for $d>1$ in \cite{ADM}. For $d=1$, the limit $\tau\to 0$ corresponds to spinless free Fermions in two dimensions in a rotating trap \cite{BSG}, with angular speed near some critical value, but it is not clear if such a physical interpretation exists for $d> 1$ for \eqref{eq:defKernelGinUEd}. The corresponding DPP can be interpreted as a higher dimensional version of the complex Ginibre ensemble.  At least formally, the approach in \cite{BSG} can be extended to $2d$ dimensions, by generalizing the Hamiltonian in a straightforward way, the model defined via \eqref{eq:defKernelGinUEd} then describing spinless free Fermions in a rotating $2d$-dimensional space (with angular speed close to some critical value). 

Contrary to the $d=1$ case, it is considerably more complicated to write down the JPDF associated to DPPs with kernel \eqref{eq:defKnMoreGeneralB} for $d>1$. For example, in what is probably the simplest case, i.e., the kernel \eqref{eq:defKernelGinUEd}, writing down the JPDF merely for $n=2$ gives 
\begin{align*}
\rho_2^{(d+1)}(z_{(1)}, \ldots, z_{(d+1)}) &= \frac{e^{-\sum_{j=1}^{d+1}|z_{(j)}|^2}}{\pi^{d(d+1)}}
\det \left(1+z_{(i)} \cdot z_{(j)}\right)_{1\leq i,j\leq d+1}
, & z_{(1)}, \ldots, z_{(d+1)}\in\mathbb C^d. 
\end{align*}
A Vandermonde type expression, with factors expressing the mutual distances between the points, is therefore not going to work, because the total degree of such a factor would be $d(d+1)$ rather than $2d+2$, and we have little hope that a closed form expression can be derived for general $n$. Consequently, it is not easy to find an explicit equilibrium problem for measures on $\mathbb C^d$ corresponding to the DPP. An associated equilibrium problem for measures on $\mathbb C^d$ can be found in \cite{Berman}, although its characterisation is somewhat abstract. For the particular cases of \eqref{eq:defKernelGinUEd} and \eqref{eq:defKerneleGinUEd}, it is implied by \cite[Theorem 3.4]{Berman} that we should find a uniform law, given by the Monge-Ampère measure \cite{Kolodziej}, although it is not a priori clear what the support is. Nevertheless, it is possible to determine the support \cite{ADM}. Under a scaling $(z, w)\mapsto \sqrt n (z, w)$, the average density of particles for both \eqref{eq:defKernelGinUEd} and \eqref{eq:defKerneleGinUEd} converges to a uniform law on a $2d$-dimensional ellipsoidal region, given by
\begin{align} \label{eq:defBulk}
\mathscr E_\tau^d = \left\{z\in\mathbb C^d : \frac{1-\tau}{1+\tau} |\Re z|^2 + \frac{1+\tau}{1-\tau} |\Im z|^2 < 1\right\}
\end{align}
(For $\tau=0$, we prove this in Corollary \ref{cor:densityd0}). The droplet is thus given by $\mathscr E_\tau^d\cup \partial\mathscr E_\tau^d$.
In the $2d$-dimensional bulk, defined via \eqref{eq:defBulk}, there is a local scaling limit in the form of a factorization in Ginibre kernels \cite{ADM}. That is, there exist unimodular functions $c_n : \mathscr E_\tau^d \times \mathbb C^d\to \mathbb T$ such that 
\begin{align*}
\lim_{n\to\infty} c_n(z, u) \overline{c_n(z, v)} \mathscr K_n\left(\sqrt n \, z+ u, \sqrt n \, z+ v\right)
= \frac{1}{\pi^d} \exp\left(u\cdot v - \frac{|u|^2+|v|^2}{2}\right),
\end{align*}
uniformly on compact sets of $z\in \mathscr E_\tau^d$ and $u, v\in\mathbb C$ (for $\tau=0$, see Corollary \ref{cor:bulkScalingLimit}). Here $z\cdot w = z_1 \overline w_1+\ldots+z_d \overline w_d$ denotes the dot product. This bulk scaling limit is well-known to be universal for random normal matrices under certain conditions ($d=1$) \cite{AmHeMa}. For $d> 1$, the bulk scaling limit is also universal, as was shown by Berman \cite{BermanFermion}. A scaling limit at the edge was not obtained in \cite{ADM} (not for fixed $\tau$). The aim of this paper is to prove that we find edge scaling limits, similar to \eqref{eq:densityLeeRiser} and \eqref{eq:FaddeevaRiserByunEbke}, for the DPPs with kernel \eqref{eq:defKernelGinUEd} or \eqref{eq:defKerneleGinUEd} on the hyperellipsoid $\partial \mathscr E_\tau^d$. 
Indeed, we find higher dimensional generalizations. We start with a result for the average density of points $\rho_n^{(1)}(z)$ (for   \eqref{eq:defKernelGinUEd} when $\tau=0$, and \eqref{eq:defKerneleGinUEd} when $0<\tau<1$). 

\begin{theorem}\label{thm:mainThmDensity}
Let $d$ be a positive integer, let $0\leq \tau<1$, and let $0<\nu<\frac{1}{6}$. Let $z\in \partial \mathscr E_\tau^d$, and denote by $\textbf{n}$ the outward unit normal vector on $\partial \mathscr E_\tau^d$ at $z$. Then we have as $n\to\infty$ that
\begin{align} \label{eq:densityScalingLimit}
n^d \, \rho_n^{(1)}\left(\sqrt n \, z+\lambda \textbf{n}\right) = \frac{d!}{2 \pi^d} \operatorname{erfc}\left(\sqrt 2 \, \lambda\right)
+ \frac{\kappa}{\sqrt n}  \frac{d!}{3\pi^d\sqrt{2\pi}}
\left(\lambda^2-1
+ \mathfrak{1}_{\tau\neq 0} \frac{d-1}{(2\kappa^2)^{\frac{1}{3}}} +\mathcal O(n^{-\frac{1}{2}+3\nu})\right) e^{-2\lambda^2},
\end{align}
uniformly for $z\in \partial\mathscr E_\tau^d$, and $\lambda\in\mathbb R$ such that $\lambda = \mathcal O(n^\nu)$, where $\kappa=\kappa(z)$ is defined by
\begin{align} \label{eq:defkappad}
\kappa(z) = \frac{1}{\left((|\Re z|^2-|\Im z|^2 - \frac{4\tau}{1-\tau^2})^2+ 4 |\Re z|^2|\Im z|^2\right)^\frac{3}{4}}.
\end{align} 
\end{theorem}
This result should be seen as an extension of \eqref{eq:densityLeeRiser} to higher dimensions (although we give fewer terms in the expansion). This is not entirely surprising. The asymptotic behavior in terms of the complementary error function was also found near the edge in higher dimensional models investigated by Ross-Singer \cite{RoSi} and Zelditch-Zhou \cite{ZeZh}. Both references investigate Bergman kernels, or rather their associated densities, in the geometric setting of complex manifolds with holomorphic line bundles, Kähler manifolds in particular. In \cite{RoSi}, the tail of the Bergman kernel is considered, this situation is different from ours but nevertheless yields the complementary error function behavior. 
We suspect that \eqref{eq:densityScalingLimit} can be derived to leading order from \cite{ZeZh} after making the relevant identifications.  
It is an interesting question whether the universality class contains other models of the form \eqref{eq:defKnMoreGeneralB}. 
At the moment, it is not known whether, analogous to the $d=1$ case, the expression $\kappa$ in \eqref{eq:defkappad} has a geometric interpretation pertaining to $\partial\mathscr E_\tau^d$ (although obviously, it yields the curvature of $\partial \mathscr E_\tau$ in $|\Re z|+i|\Im z|$).\\


We also derive a scaling limit, essentially a higher dimensional analogue of the Faddeeva plasma kernel, for the kernel with general arguments (i.e., not necessarily diagonal or in the direction of the outward normal vector) near the edge. 
\begin{theorem}\label{thm:mainThmFaddeevaKernel}
Let $d$ be a positive integer, let $0\leq\tau<1$, and let $0<\nu<\frac{1}{6}$.  Let $z\in \partial \mathscr E_\tau^d$, and denote by $\textbf{n}$ the outward unit normal vector on $\partial \mathscr E_\tau^d$ at $z$. Then there exist continuous unimodular functions $c_n : \partial \mathscr E_\tau^d \times \mathbb C^d\to \mathbb T$ such that 
\begin{multline*}
c_n(z, u) \overline{c_n(z, v)} \pi^d \exp\left(\frac{|u|^2+|v|^2}{2}-u\cdot v\right) \mathscr K_n(\sqrt n \, z+u, \sqrt n \, z+v)
= \frac{1}{2} \operatorname{erfc}\left(\frac{u\cdot \textbf{n}+\textbf{n}\cdot v}{\sqrt 2}\right) \\
 +  \exp\left(- \frac{(u\cdot \textbf{n}+\textbf{n}\cdot v)^2}{2}\right) \frac{\mathcal O\left(1+|u|^2+|v|^2\right)}{\sqrt n},
\end{multline*}
as $n\to\infty$, uniformly for $z\in \partial\mathscr E_\tau^d$, and $u, v\in\mathbb C^d$ such that $|u|, |v| = \mathcal O(n^\nu)$.
\end{theorem}

For $d=1$, we can be more explicit about the error, and indeed, we manage to rederive a slightly stronger version of \eqref{eq:FaddeevaRiserByunEbke} (see Proposition \ref{prop:d=1FaddeevaSlightlyStronger}). We consider our derivation of this result of independent interest. 
For general $d\geq 1$, the scaling limit is expressed in terms of the geometry of the droplet, i.e., in terms of the unit normal vector. We can directly extract a weaker result where this is not the case, where the scaling limit becomes the Faddeeva plasma kernel. 

\begin{corollary} \label{cor:Fadd}
Under the conditions of Theorem \ref{thm:mainThmFaddeevaKernel}, there exist continuous unimodular functions $c_n : \partial \mathscr E_\tau^d \times \mathbb C^d\to \mathbb T$ such that 
\begin{align*}
\lim_{n\to\infty} c_n(z, u \, \textbf{n}) \overline{c_n(z, v \, \textbf{n})} \mathscr K_n(\sqrt n \, z+u \, \textbf{n}, \sqrt n \, z+v \, \textbf{n}) = \frac{1}{2 \pi^d} \exp\left(\frac{|u|^2+|v|^2}{2}-u \overline v\right) 
\operatorname{erfc}\left(\frac{u+\overline{v}}{\sqrt 2}\right) ,
\end{align*}
uniformly for $z\in \partial\mathscr E_\tau^d$, and $u, v\in\mathbb C$ such that $u, v = \mathcal O(n^\nu)$.
\end{corollary}
Indeed, this result shows that the universality of  \eqref{eq:edgeScalingErfcGinibreRewrite} is not limited to the eigenvalues of random normal matrices and other non-Hermitian random matrices, but its universality class includes higher dimensional models as well. To the best of our knowledge, a scaling limit of this form (excluding the diagonal case) has not yet appeared in the literature for higher dimensional models.  While we cannot claim that the elliptic class is general enough to imply a universality theorem for generic weights $\mathscr W$, our results do provide a target for such a result. One possible direction to pursue is as follows. It is known that a version of the Mehler kernel exists for generalized Laguerre polynomials, expressed via the so-called Hardy-Hille formula, and an adaptation of the approach of the current paper and \cite{ADM} will likely work here. Our approach likely also works for distinct eccentricity  parameters, i.e., we may consider $\mathscr W(z) = \omega_1(z_1) \cdots \omega_d(z_d)$ with $\omega_k(z) = e^{-|z|^2+\tau_k \operatorname{Re}(z^2)}$ and $\tau_1,\ldots, \tau_d\in (0,1)$. This case can be treated with a single integral representation as well, and likely leads to a similar situation where a saddle point and pole coalesce in the limit $n\to\infty$, which, as we shall see, produces the complementary error function. In the end, a \textit{universality theorem} that we wish to prove in a future work, is that some version of Corollary \ref{cor:Fadd} holds for a general class of multivariate weights $\mathscr W(z) = e^{-n \mathscr V(z)}$ on $\mathbb C^d$, much larger than our restricted class of factorized weights with external field $\mathscr V(z) = |z|^2 - \tau \operatorname{Re}(z_1^2+\ldots+z_d^2)$ (in that case, we should change $\sqrt n z+u$ and $\sqrt n z+v$ to $z+\frac{u}{\sqrt n}$ and $z+\frac{v}{\sqrt n}$ due to the different scaling).\\ 



The paper is built up as follows. In Section \ref{sec:set-up} we recap some results from \cite{ADM}, that where used for a steepest descent analysis to derive the asymptotic behavior of \eqref{eq:defKerneleGinUEd}. The reason that the edge limit was not treated in \cite{ADM}, is that it corresponds to a more complicated situation where a saddle point and pole coalesce in the limit $n\to\infty$. In Section \ref{sec:steepest} we clarify how this situation can be treated. Some preparatory relevant identities and estimates are derived in Section \ref{sec:relevant}, and finally, in Section \ref{sec:proofs} we prove the main theorems.\\ 





\section*{Acknowledgments}
The author is funded by the Deutsche Forschungsgemeinschaft (DFG, German Research Foundation) – SFB 1283/2 2021
– 317210226 "Taming uncertainty and profiting from randomness and low regularity in analysis, stochastics and their applications", and the Royal Society grant RF\textbackslash ERE\textbackslash 210237. The author thanks Gernot Akemann, Sung-Soo Byun, Maurice Duits, Markus Ebke, Ivan Parra and Roman Riser for useful discussions.

%

\section{Approach and set up} \label{sec:set-up}

\subsection{Single integral representation}

Our primary tool to prove the main results will be a steepest descent analysis. In \cite{ADM} it was shown, using a formula for the Mehler kernel, that \eqref{eq:defKerneleGinUEd} (with a different scaling) admits a single integral representation of the form
\begin{align} \label{eq:defKnInIn}
\mathscr K_n(\sqrt n \, z, \sqrt n \, w) = \left(\frac{\sqrt{1-\tau^2}}{\pi}\right)^d 
\sqrt{ \omega\left(\sqrt n \, z_+\right)  \omega\left(\sqrt n \, z_-\right)} \, I_{n,\tau}^d(z_\pm),
\end{align}
valid for any $z, w\in\mathbb C^d$, where
\begin{align} \label{eq:defz+-}
z_\pm = \frac{\sqrt{\sinh 2\xi_\tau}}{2} \left(\sqrt{\sum_{j=1}^d (z_j+\overline w_j)^2} \pm \sqrt{\sum_{j=1}^d (z_j-\overline w_j)^2}\right),
\end{align}
and, with $\gamma_0$ a small positively oriented loop around $s=0$, the integral $I_{n,\tau}^d(z_\pm)$ is given by
\begin{align} \label{eq:defIntaud}
I_{n,\tau}^d(z_\pm) = -\frac{1}{2\pi i} \oint_{\gamma_0} \frac{e^{n F(s)}}{s-\tau} \frac{ds}{(1-s^2)^\frac{d}{2}}
\end{align}
and $F(s) = F_\tau(z_\pm;s)$  is given by
\begin{align} \label{eq:defF}
F(s) = \frac{s}{1+s} \frac{(z_++z_-)^2}{2}-\frac{s}{1-s} \frac{(z_+-z_-)^2}{2} - \log s + \log \tau.
\end{align}
We do not have to be explicit about the choice of branch of the square roots in \eqref{eq:defz+-}, because the choice is irrelevant for \eqref{eq:defF}, and thus for $\mathscr K_n(\sqrt n \, z, \sqrt n \, w)$. The same holds for the branch of the logarithm in \eqref{eq:defF}. Nevertheless, let us use the convention that logarithms and power functions will be defined as $\log z = \log|z|+i\arg z$ and $z^\alpha = |z|^\alpha e^{i\alpha\arg z}$, where $\arg z\in (-\pi, \pi]$. 
The saddle points of $F(s)$ have a surprisingly simple form in elliptic coordinates. We write
\begin{align*}
z_\pm = \sqrt 2 \cosh(\xi_\pm + i \eta_\pm),
\end{align*}
where $\xi_\pm \geq 0$ and $\eta_\pm\in (-\pi,\pi]$ when $\xi_\pm>0$, while $\eta_\pm\in [0,\pi]$ when $\xi_\pm=0$. Note that any constant value of $\xi_\pm$ corresponds to an ellipse with vertex $\sqrt 2 \cosh \xi_\pm$ and co-vertex $\sqrt 2 \sinh \xi_\pm$.
The edge, i.e., the ellipse forming the boundary of the droplet, is described by the particular choice $\xi_\pm = \xi_\tau$, where $\xi_\tau = \frac{1}{2} \log \frac{1}{\tau}$. We define
\begin{align} \label{eq:defab}
a &= e^{\xi_++\xi_-+i(\eta_++\eta_-)}, \quad \text{ and } \quad b = e^{\xi_+-\xi_-+i(\eta_+-\eta_-)}.
\end{align}
With these notations, the following result was derived in \cite{ADM}.

\begin{proposition} [$0<\tau<1$] \label{prop:saddlePointsDef}
\text{ }\\
If $z_+, z_- \in \mathbb C\setminus \{-\sqrt 2, \sqrt 2\}$, then the saddle points of $s\mapsto F_\tau(z_\pm;s)$ are simple, and we have the following:
\begin{itemize}
\item[(i)] When $z_+ \neq  \pm z_-$, there are exactly four saddle points given by $a, a^{-1}, b$ and $b^{-1}$.
\item[(ii)] When $z_+=\pm z_-$ and $z_+\neq 0$, there are exactly two  saddle points, which are given by $a$ and $a^{-1}$.
\end{itemize}
If $z_+\in \{-\sqrt{2},\sqrt{2}\}$ or $z_-\in \{-\sqrt{2},\sqrt{2}\}$, then all saddle points have order two and we have the following:
\begin{itemize}
    \item[(iii)] When $z_\pm \in \{-\sqrt 2, \sqrt 2\}$ and $z_\mp\not\in \{-\sqrt 2, \sqrt 2\}$, then we have two saddle points $a=b^{\mp 1}$ and $a^{-1}=b^{\pm 1}$.
    \item[(iv)] When $z_+=\pm z_-\in\{-\sqrt 2, \sqrt 2\}$, then we have one saddle point $a^{-1}=a=b=b^{-1}=\pm 1$.
\end{itemize}
Finally, when $z_+=z_-=0$ there are no saddle points.
\end{proposition}

Furthermore, in order to understand what deformations of $\gamma_0$ were allowed, the following theorem was proved.
Notice in particular that one can deform $\gamma_0$ to the circle $|s|=|a|^{-1}$ when $\xi_+, \xi_->0$. 

\begin{theorem}[$0<\tau<1$] \label{lem:as=1}
With the notations as above, we have the inequality
\begin{align} \label{eq:as=1}
\operatorname{Re} F(s) &\leq \operatorname{Re} F(a^{-1}), & |s| = |a|^{-1}.
\end{align} 
\begin{itemize}
\item[(i)] When $\xi_+>0$ and $\xi_->0$, we have equality if and only if $s = a^{-1}$.
\item[(ii)] When $\xi_+>0$ and $\xi_-=0$, we have equality if and only if $s = a^{-1}$ or $s=b^{-1}$. 
\item[(iii)] When $\xi_+=0$ and $\xi_->0$, we have equality if and only if $s = a^{-1}$ or $s=b$.
\item[(iv)] When $\xi_+=0$ and $\xi_-=0$, we have equality for all $s$. 
\end{itemize}
\end{theorem}

So far, these results describe $\mathscr K_n$ for $0<\tau<1$ (as defined in \eqref{eq:defKerneleGinUEd}). For $\tau=0$, we can significantly simplify $\mathscr K_n$ (as defined in \eqref{eq:defKernelGinUEd}). 

\begin{proposition}[$\tau=0$] \label{prop:KnisKnford=1}
Let $z, w\in \mathbb C^d$. Then we may write
\begin{align} \label{eq:defKnInIn0}
\mathscr K_n(\sqrt n \, z, \sqrt n \, w)
= \frac{1}{\pi^d} e^{-n \frac{|z|^2+|w|^2}{2}} \sum_{j=0}^{n-1} \frac{(n \, z\cdot w)^j}{j!}. 
\end{align}
\end{proposition}

\begin{proof}
We start by noticing that
\begin{align*}
e^{(z\cdot w) s} 
= \prod_{k=1}^d e^{(z_k \overline w_k) s} 
= \prod_{k=1}^d \sum_{j_k=0}^\infty \frac{(z_k \overline w_k s)^{j_k}}{j_k!}
= \sum_{m=0}^\infty \sum_{|j|=m} s^m \prod_{k=1}^d \frac{(z_k \overline w_k)^{j_k}}{j_k!}.
\end{align*}
Then by applying the residue theorem in two directions, we have
\begin{align} \label{eq:residueTrickIntegral}
\sum_{|j|<n} \prod_{k=1}^d \frac{(z_k \overline w_k)^{j_k}}{j_k!}
= \frac{1}{2\pi i} \oint_{\gamma_0} e^{(z\cdot w) s} \left(1+\frac{1}{s}+\ldots+\frac{1}{s^{n-1}}\right) \frac{ds}{s}
= \sum_{j=0}^{n-1} \frac{(z\cdot w)^j}{j!}. 
\end{align}
Now substituting $(z, w) \mapsto \sqrt n (z, w)$, and multiplying with the remaining factors in \eqref{eq:defKernelGinUEd}, we arrive at \eqref{eq:defKnInIn0}. 
\end{proof}

Comparing with \eqref{eq:defGinibreKer}, we infer that $\mathscr K_n$ for $d\geq 1$ and $\tau=0$ is closely related to the Ginibre ensemble ($d=1$). In particular, some properties are immediately inherited from the Ginibre ensemble. For instance, the model has a bulk, which is given by the $2d$-dimensional unit ball in $\mathbb C^d$. 

\begin{corollary}[$\tau=0$] \label{cor:densityd0}
The average density of particles converges to a uniform law on the unit ball in $\mathbb C^d$.\\
More precisely, we have
\begin{align} \label{eq:densityGinibreScalingd0}
\lim_{n\to\infty} n^d \, \rho_n^{(1)}(\sqrt n \, z)
= \begin{cases} \displaystyle \frac{d!}{\pi^d}, & |z|<1,\\ \displaystyle \frac{d!}{2\pi^d}, & |z|=1,\\ \displaystyle 0, & |z|>1.\end{cases} 
\end{align}
The convergence is uniform on compact subsets such that $|z|\neq 1$. 
\end{corollary}

\begin{proof}
By some straightforward combinatorial arguments, the number of points of the DPP defined via \eqref{eq:defKernelGinUEd} is given by the binomial coefficient $\binom{n+d-1}{d}$, which behaves as $\frac{n^d}{d!} (1+\mathcal O(1/n))$ for large $n$.
For the average density of particles, one takes $z=w$. Then \eqref{eq:defKnInIn0} turns into
\begin{align*} 
\mathscr K_n(\sqrt n \, z, \sqrt n \, z)
= \frac{1}{\pi^d} e^{-n |z|^2} \sum_{j=0}^{n-1} \frac{(n |z|^2)^j}{j!}. 
\end{align*}
This, apart from a factor $\pi^{d}$ rather than $\pi$, is simply the (rescaled) average density of points of the Ginibre ensemble ($d=1$) in $|z|$ (or any rotation in the plane thereof), and we obtain the result directly from \eqref{eq:densityGinibreScaling}. 
\end{proof}

Though Proposition \ref{prop:KnisKnford=1} gives an immediate relation between the $d>1$ and $d=1$ model for $\tau=0$, it will turn out to be both instructive and beneficial to find a single integral representation for this model as well. We have the following result. 

\begin{proposition} [$\tau=0$] \label{prop:defFd0}
Let $z, w\in \mathbb C^d$. We may write
\begin{align} \label{eq:defKnInIn0B}
\mathscr K_n(\sqrt n \, z, \sqrt n \, w)
= \frac{1}{\pi^d} e^{-n \frac{|z|^2+|w|^2}{2}} I_{n,0}^d(z\cdot w), 
\end{align}
where, for $\zeta\in\mathbb C$, we have
\begin{align} \label{eq:defInod}
I_{n,0}^d(\zeta) = - \frac{1}{2\pi i} \oint_{\gamma_0} \frac{e^{n F(s)}}{s-1} ds,
\end{align}
with $\gamma_0$ a small positively oriented loop around $0$, and $F(s) = F(\zeta; s)$ is defined by
\begin{align} \label{eq:defF0}
F(s) = \zeta s - \log s. 
\end{align}
\end{proposition}

\begin{proof}
Since $s=1$ is not enclosed by $\gamma_0$ (which is assumed to be small), the residue theorem implies that \eqref{eq:residueTrickIntegral} can alternatively be written as
\begin{align} \label{eq:residueTrickIntegral2}
\sum_{|j|<n} \prod_{k=1}^d \frac{(z_k \overline w_k)^{j_k}}{j_k!}
= \frac{1}{2\pi i} \oint_{\gamma_0} e^{(z\cdot w) s} \frac{1-s^{-n}}{1-s^{-1}} \frac{ds}{s}
= -\frac{1}{2\pi i} \oint_{\gamma_0} e^{(z\cdot w) s} s^{-n} \frac{ds}{s-1}.
\end{align}
Now substituting $(z, w) \mapsto \sqrt n (z, w)$, writing $s^{-n} = e^{-n\log s}$, and multiplying with the remaining factors in \eqref{eq:defKernelGinUEd}, we arrive at \eqref{eq:defKnInIn0B}, with $I_{n,0}^d$ and $F$ as defined in \eqref{eq:defInod} and \eqref{eq:defF0} respectively. 
\end{proof}

We state the following proposition. The proof is trivial, and is therefore omitted. 

\begin{proposition}[$\tau=0$] \label{prop:saddlePointstau0d}
Let $\zeta\in\mathbb C$, and let $F(s) = F(\zeta; s)$ be as in Proposition \ref{prop:defFd0}, i.e., $F(s) = \zeta s - \log s$.
\begin{itemize}
\item[(i)] When $\zeta\neq 0$, there is only one saddle point $s_0=\zeta^{-1}$, which is simple.\\
In this case $F(s_0) = 1+\log \zeta$ and $F''(s_0) = \zeta^2$.
\item[(ii)] When $\zeta=0$, there are no saddle points.
\end{itemize}
\end{proposition}

\subsection{Set up}

We shall need the results of the previous subsection, tailored to a specific situation. In our present case, we investigate the asymptotic behavior of $\mathscr K_n(\sqrt n \, z + u, \sqrt n \, z + v)$, where $z\in \partial \mathscr E_\tau^d$, and $u, v\in\mathbb C^d$ satisfy $u, v = \mathcal O(n^\nu)$ as $n\to\infty$, for some fixed $0<\nu<\frac{1}{6}$. First, we focus on the case $0<\tau<1$. Using the definition in \eqref{eq:defz+-}, we have that
\begin{align} \label{eq:z+-OurCase}
\frac{z_\pm}{\sqrt{\sinh 2\xi_\tau}} = \sqrt{|\Re z|^2+\frac{u+\overline v}{\sqrt n}\cdot (\Re z)+\frac{(u+\overline v)^2}{4 n}}
\pm i \sqrt{|\Im z|^2 + \frac{u-\overline v}{\sqrt n}\cdot (i\Im z)-\frac{(u-\overline v)^2}{4 n}},
\end{align}
where we use a short-hand notation $\zeta^2 = \zeta_1^2+\ldots+\zeta_d^2$ for any $\zeta\in\mathbb C^d$. In particular, $z_\pm$ depend piecewise continuously on $z$ and $\frac{u}{\sqrt n}, \frac{v}{\sqrt n}$, and we have uniformly for $z\in \partial \mathscr E_\tau^d$ and $u, v=\mathcal O(n^\nu)$ that
\begin{align*}
\lim_{n\to\infty} \frac{z_\pm}{\sqrt{\sinh 2\xi_\tau}} = \frac{\hat z_\pm}{\sqrt{\sinh 2\xi_\tau}} :=  |\Re z|\pm i|\Im z| \in \partial \mathscr E_\tau.
\end{align*}
In particular, there is an $\eta\in [0,\frac{\pi}{2}]$ such that
\begin{align*}
\hat z_\pm = \sqrt 2 \cosh(\xi_\tau\pm i\eta). 
\end{align*}
When $d=1$, we simply have $z_+ = \sqrt{\sinh 2\xi_\tau} \, (z+u)$ and $z_- = \sqrt{\sinh 2\xi_\tau} \, \, \overline{(z+v)}$. When $d>1$, there is an essential difference though, and this is caused by the branch-cut of the square roots in \eqref{eq:z+-OurCase}. In the end, these branch-cuts will somehow drop out (at least to the dominant order), since $I_{n,\tau}^d(z_\pm)$ depends only on $(z_+\pm z_-)^2$, which have no branch-cut. However, in our derivation we shall have to take the presence of these branch-cuts into account. When $z$ is in a subset of $\partial \mathscr E_\tau^d$ such that $\Re z, \Im z\neq 0$, we have
\begin{align} \label{eq:z+-asympn}
\frac{z_\pm}{\sqrt{\sinh 2\xi_\tau}} &= \frac{\hat z_\pm}{\sqrt{\sinh 2\xi_\tau}} +  \frac{u+\overline v}{\sqrt n} \cdot \frac{\Re z}{|\Re z|}  \pm \frac{u-\overline v}{\sqrt n} \cdot\frac{\Im z}{|\Im z|}  
+ \mathcal O\left(\frac{n^{-1+2\nu}}{|\Re z|}+\frac{n^{-1+2\nu}}{|\Im z|}\right)
\end{align} 
uniformly for $u, v = \mathcal O(n^\nu)$ as $n\to\infty$.
On the other hand, when $\Re z = 0$, we have
\begin{align} \label{eq:z+-asympnRez=0}
\frac{z_\pm}{\sqrt{\sinh 2\xi_\tau}} &= \frac{\hat z_\pm}{\sqrt{\sinh 2\xi_\tau}} +  \frac{\sqrt{(u+\overline v)^2}}{\sqrt n}  \pm \frac{u-\overline v}{\sqrt n} \cdot\frac{\Im z}{|\Im z|}  
+ \mathcal O\left(n^{-1+2\nu}\right),
\end{align} 
while the case $\Im z=0$ yields
\begin{align} \label{eq:z+-asympnImz=0}
\frac{z_\pm}{\sqrt{\sinh 2\xi_\tau}} &= \frac{\hat z_\pm}{\sqrt{\sinh 2\xi_\tau}} +  \frac{u+\overline v}{\sqrt n} \cdot \frac{\Re z}{|\Re z|}  \pm \frac{\sqrt{(u-\overline v)^2}}{\sqrt n}
+ \mathcal O\left(n^{-1+2\nu}\right),
\end{align} 
uniformly for $u, v = \mathcal O(n^\nu)$ as $n\to\infty$. Although we have $z_\pm = \hat z_\pm + \mathcal O(n^{-\frac{1}{2}+\nu})$ in all cases, it is not a priori clear that the constant implied by this $\mathcal O$ term can be taken uniformly for $z\in \partial \mathscr E_\tau^d$. Nevertheless, this does indeed hold. 

\begin{lemma}[$0<\tau<1$]
We have $z_\pm-\hat z_\pm = \mathcal O(n^{-\frac{1}{2}+\nu})$ as $n\to\infty$, uniformly for $z\in \partial\mathscr E_\tau^d$ and $u,v=\mathcal O(n^\nu)$. 
\end{lemma}

\begin{proof}
One simply has to divide into the case $|\Re z|, |\Im z|> n^{-\frac{1}{2}+\nu}$, the case $|\Re z|\leq n^{-\frac{1}{2}+\nu}$, and the case $|\Im  z|\leq n^{-\frac{1}{2}+\nu}$. In the first case, we have
\begin{align} \nonumber
\left|\sqrt{|\Re z|^2+\frac{u+\overline v}{\sqrt n}\cdot (\Re z)+\frac{(u+\overline v)^2}{4 n}} - |\Re z|\right|
&=  \left|\frac{\frac{u+\overline v}{\sqrt n}\cdot \frac{\Re z}{|\Re z|}+\frac{(u+\overline v)^2}{4 |\Re z| n}}{1+\sqrt{1+\frac{u+\overline v}{\sqrt n}\cdot \frac{\Re z}{|\Re z|}+\frac{(u+\overline v)^2}{4 |\Re z| n}}}\right|
\leq \frac{|u+\overline v|}{\sqrt n} + \frac{|u+\overline v|^2}{4 n^{\frac{1}{2}+\nu}},\\ \label{eq:absSqrtImz}
\left|\sqrt{|\Im z|^2+\frac{u-\overline v}{\sqrt n}\cdot (i\Im z)-\frac{(u-\overline v)^2}{4 n}} - |\Im z|\right|
&= \left|\frac{\frac{u-\overline v}{\sqrt n}\cdot \frac{i\Im z}{|\Im z|}-\frac{(u-\overline v)^2}{4 |\Im z| n}}{1+\sqrt{1+\frac{u+\overline v}{\sqrt n}\cdot \frac{\Re z}{|\Re z|}+\frac{(u+\overline v)^2}{4 |\Re z| n}}}\right|
\leq \frac{|u-\overline v|}{\sqrt n} + \frac{|u-\overline v|^2}{4 n^{\frac{1}{2}+\nu}}.
\end{align}
We used here that for any $\zeta\in\mathbb C$ we have $\Re \sqrt \zeta\geq 0$, and thus for all $\zeta\in\mathbb C\cup \{\infty\}$ one has
\begin{align*}
\left|\frac{1}{1+\sqrt \zeta}\right|\leq 1.
\end{align*}
When we are in the second case, we have
\begin{align*}
\left|\sqrt{|\Re z|^2+\frac{u+\overline v}{\sqrt n}\cdot (\Re z)+\frac{(u+\overline v)^2}{4 n}} - |\Re z|\right|
&\leq \sqrt{n^{-1+2\nu}+|u+\overline v| n^{-1+\nu}+\frac{|u+\overline v|^2}{4 n}} + n^{-\frac{1}{2}+\nu}
= n^{-\frac{1}{2}+\nu} \left(2+\frac{|u+\overline v|}{2 n^\nu}\right),
\end{align*}
while we still have the same estimate \eqref{eq:absSqrtImz}. The third case is analogous. Obviously, we can find a constant $C>0$ such that $|z_\pm-\hat z_\pm|\leq C n^{-\frac{1}{2}+\nu}$ for the three cases simultaneously. 
\end{proof}

It will turn out to be convenient to introduce
\begin{align} \label{eq:defDelta+-}
\Delta_\pm = \frac{z_\pm - \hat z_\pm}{\sqrt{\hat z_\pm^2-2}} = \frac{\cosh(\xi_\pm+i\eta_\pm) - \cosh(\xi_\tau\pm i\eta)}{\sinh(\xi_\tau\pm i\eta)},
\end{align}
where, as before, $z_\pm = \sqrt 2 \cosh(\xi_\pm+i\eta_\pm)$ and $\hat z_\pm = \sqrt 2 \cosh(\xi_\tau\pm i\eta)$. 
Rather than working with $u, v=\mathcal O(n^\nu)$ we shall simply consider $z_\pm = \hat z_\pm + \sqrt{\hat z_\pm^2-2} \, \Delta_\pm$, and demand that $\Delta_\pm = \mathcal O(n^{-\frac{1}{2}+\nu})$ as $n\to\infty$. An incidental advantage of this perspective, is that we only need to consider limits that are uniform in $\hat z_\pm \in \sqrt{\sinh 2\xi_\tau} \, \partial \mathscr E_\tau = \{\sinh(2\xi_\tau) \, w : w\in \partial \mathscr E_\tau\}$. For notational convenience, we introduce 
\begin{align*}
\hat z = \frac{z_+}{\sqrt{\sinh 2\xi_\tau}}.
\end{align*}
Summarizing, our aim is to understand the asymptotic behavior of $I_{n,\tau}^d(z_\pm)$, where
\begin{align*}
z_\pm = \hat z_\pm+\sqrt{\hat z_\pm^2-2} \, \Delta_\pm,
\end{align*}
uniformly for $\hat z \in \partial \mathscr E_\tau$, under the assumption that $\Delta_\pm=\mathcal O(n^{-\frac{1}{2}+\nu})$ as $n\to\infty$. 

\begin{lemma}[$0<\tau<1$] \label{lem:abAsympDelta}
Let $a$ and $b$ be defined as in \eqref{eq:defab}. Then we have
\begin{align} \label{eq:a+-1Deltas}
a^{\pm 1} &= \tau^{\mp 1} 
\pm \tau^{\mp 1}(\Delta_++\Delta_-)
+ \mathcal O(n^{-1+2\nu}).\\ \label{eq:b+-1Deltas}
b^{\pm 1} &= e^{\pm 2 i \eta} 
\pm e^{\pm 2 i \eta}  (\Delta_+-\Delta_-) 
+ \mathcal O(n^{-1+2\nu}).
\end{align}
as $n\to\infty$, uniformly for $\hat z \in \partial \mathscr E_\tau$. Furthermore, we have 
\begin{align} \label{eq:atoHighAccuracy}
\tau a - 1 = \Delta_++\Delta_-+\frac{(\Delta_++\Delta_-)^2}{2}
- \coth(\xi_\tau+i\eta) \frac{\Delta_+^2}{2} 
- \coth(\xi_\tau-i\eta) \frac{\Delta_-^2}{2}
+ \mathcal O(n^{-\frac{3}{2}+3\nu})
\end{align}
as $n\to\infty$, uniformly for $\hat z \in \partial \mathscr E_\tau$.
\end{lemma}

\begin{proof}
As shown in \cite{ADM}, we can alternatively write the saddle points as
\begin{align} \label{eq:a+-1z+z-}
a^{\pm 1} &= \frac{1}{2} \left(z_+\pm \sqrt{z_+^2-2}\right) \left(z_-\pm\sqrt{z_-^2-2}\right)\\
b^{\pm 1} &= \frac{1}{2} \left(z_+\pm \sqrt{z_+^2-2}\right) \left(z_-\mp\sqrt{z_-^2-2}\right).
\end{align}
We observe by Taylor series expansion that
\begin{align*}
\sqrt{z_\pm^2-2} &= \sqrt{\hat z_\pm^2-2} + \frac{\hat z_\pm}{\sqrt{\hat z_\pm^2-2}} (z_\pm-\hat z_\pm)
- \frac{1}{(\hat z_\pm^2-2)^\frac{3}{2}} (z_\pm - \hat z_\pm)^2 + \mathcal O(\Delta_\pm^3)\\
&= \sqrt{\hat z_\pm^2-2} + \hat z_\pm \Delta_\pm
- \frac{\Delta_\pm^2}{\sqrt{\hat z_\pm^2-2}} + \mathcal O(\Delta_\pm^3).
\end{align*}
Henceforth
\begin{align} \label{eq:zWz+-}
\frac{z_\pm+\sqrt{z_\pm^2-2}}{\sqrt 2} =
e^{\xi_\tau\pm i\eta} + e^{\xi_\tau\pm i\eta} \Delta_\pm
- \frac{1}{\sinh(\xi_\tau\pm i\eta)} \frac{\Delta_\pm^2}{2} + \mathcal O(\Delta_\pm^3).
\end{align}
Plugging these in \eqref{eq:a+-1z+z-}, we obtain
\begin{align} \label{eq:a+-1Deltas}
a^{\pm 1} = \tau^{\mp 1} 
\pm \tau^{\mp 1}(\Delta_++\Delta_-)
+ \mathcal O(n^{-1+2\nu}). 
\end{align}
Analogously, we find
\begin{align} \label{eq:b+-1Deltas}
b^{\pm 1} = e^{\pm 2 i \eta} 
\pm e^{\pm 2 i \eta}  (\Delta_+-\Delta_-) 
+ \mathcal O(n^{-1+2\nu}). 
\end{align}
Using \eqref{eq:zWz+-} to one order higher, we find that
\begin{align*} 
\tau a - 1 = \Delta_++\Delta_-+\Delta_+\Delta_-
- \frac{e^{-\xi_\tau-i\eta}}{\sinh(\xi_\tau+i\eta)} \frac{\Delta_+^2}{2} 
- \frac{e^{-\xi_\tau+i\eta}}{\sinh(\xi_\tau-i\eta)} \frac{\Delta_-^2}{2}
+ \mathcal O(n^{-\frac{3}{2}+3\nu}).
\end{align*}
\end{proof}

Lemma \ref{lem:abAsympDelta} shows in particular that $a, b$ and $b^{-1}$ lie outside the circle $|s|=|a|^{-1}$ for $n$ big enough. Following Lemma \ref{lem:as=1}, we intend to deform the integration contour $\gamma_0$ to the circle $|s|=|a|^{-1}$, and we should only take the saddle point $s=a^{-1}$ into account. What makes this situation complicated, is that the saddle point $a^{-1}$ and the pole at $\tau$ of the integrand of $I_{n,\tau}^d(\hat z_\pm + \sqrt{\hat z_\pm^2-2} \, \Delta_\pm)$ coalesce in the limit $n\to\infty$.\\

Since $a^{-1}$ is the only saddle point that will give a contribution, we shall need to know the values of $F$ and $F''$ only in $a^{-1}$. By \cite{ADM}, these are given by
\begin{align}  \nonumber
F(a^{-1}) &= 1 + \log \tau + \xi_+ + \xi_- + i(\eta_++\eta_-) +  \frac{1}{2} e^{-2 (\xi_++ i \eta_+)} + \frac{1}{2} e^{-2 (\xi_-+ i \eta_-)},\\ \label{eq:F''a-1}
F''(a^{-1}) &= 2 a^{2} \frac{\sinh(\xi_++i\eta_+) \sinh(\xi_-+i\eta_-)}{\sinh(\xi_++i\eta_++\xi_-+i\eta_-)}.
\end{align}
We would rather express $F''(a^{-1})$ in terms of $\Delta_\pm$ however, hence the following lemma. As it turns out, we need not be precise about the behavior of $F(a^{-1})$. 

\begin{lemma}[$0<\tau<1$] 
Uniformly for $\hat z \in \partial \mathscr E_\tau$, we have as $n\to\infty$ that
\begin{align} \label{eq:F''a-1Explicit}
\frac{1}{2} a^{-2} F''(a^{-1}) = \frac{|\sinh(\xi_\tau+i\eta)|^2}{\sinh(2\xi_\tau)} 
\left(1 
 + \coth(\xi_\tau+i\eta) \Delta_+
 + \coth(\xi_\tau-i\eta) \Delta_-
 - \coth(2\xi_\tau) (\Delta_++\Delta_-)\right) 
 + \mathcal O(n^{-1+2\nu}).
\end{align}
\end{lemma}

\begin{proof}
As observed in \cite{ADM}, we can write $F''(a^{-1})$ entirely in terms of the saddle points. Namely
\begin{align*}
F''(a^{-1}) = \frac{(a-b)(a-b^{-1})}{1-a^{-2}}. 
\end{align*}
Plugging in \eqref{eq:a+-1Deltas} and \eqref{eq:b+-1Deltas}, we find that
\begin{align*}
a - b^{\pm 1} &= (\tau^{-1}-e^{\pm 2i\eta}) \left(1 + \Delta_\pm + \coth(\xi_\tau\mp i\eta) \Delta_\mp\right) + \mathcal O(n^{-1+2\nu}),\\
a^2 - 1 &= (\tau^{-2}-1) \left(1 + 2\frac{\Delta_++\Delta_-}{1-\tau^2} \right) + \mathcal O(n^{-1+2\nu}).
\end{align*}
We thus have
\begin{align*}
a^{-2} F''(a^{-1}) = \frac{|\tau^{-1}-e^{2i\eta}|^2}{\tau^{-2}-1} 
\left(1 
 + \coth(\xi_\tau+i\eta) \Delta_+
 + \coth(\xi_\tau-i\eta) \Delta_-
 - \coth(2\xi_\tau) (\Delta_++\Delta_-)\right) 
 + \mathcal O(n^{-1+2\nu}).
\end{align*}
\end{proof}

So far, we set up everything for the case $0<\tau<1$. The case $\tau=0$ is considerably easier, and there is not much to set up. Substituting $\sqrt n (z, w)\mapsto (\sqrt n \, z+u, \sqrt n \, z+v)$ in Proposition \ref{prop:saddlePointstau0d}, we find that there is exactly one saddle point (for $n$ large enough), which is simple, and given by
\begin{align} \label{eq:defSaddlePoint0n}
s_0 = \frac{1}{\left(z+\frac{u}{\sqrt n}\right)\cdot \left(z+\frac{v}{\sqrt n}\right)}
= \left(1+\frac{u\cdot z+z\cdot v}{\sqrt n}+\frac{u\cdot v}{n}\right)^{-1}
= 1 - \frac{u\cdot z+z\cdot v}{\sqrt n} + \frac{(u\cdot z+z\cdot v)^2-u\cdot v}{n} + \mathcal O(n^{-\frac{3}{2}+3\nu})
\end{align}
as $n\to\infty$, uniformly for $|z|=1$ and $u, v\in\mathbb C^d$ with $u, v = \mathcal O(n^\nu)$. 
Note that, in this case, just as in the case $0<\tau<1$, the saddle point ($s=s_0$) coalescs with the pole ($s=1$) in the limit $n\to\infty$.
We also find that
\begin{align} \nonumber
F(s_0) &= 1-\log s_0 = 1 +\log \left(1+\frac{u\cdot z+z\cdot v}{\sqrt n}+\frac{u\cdot v}{n}\right)
= 1+\frac{u\cdot z+z\cdot v}{\sqrt n}-\frac{(u\cdot z+z\cdot u)^2-u\cdot v}{n} + \mathcal O(n^{-\frac{3}{2}+3\nu}),\\ \label{eq:F''ins0}
F''(s_0) &= \frac{1}{s_0^{2}} = \left(1+\frac{u\cdot z+z\cdot v}{\sqrt n}+\frac{u\cdot v}{n}\right)^2,
\end{align}
as $n\to\infty$, uniformly for $|z|=1$ and $u, v\in\mathbb C^d$ with $u, v = \mathcal O(n^\nu)$. As before, rather than using $u, v$, we shall consider the asymptotic behavior of $I_{n,0}^d(1+\Delta)$, where we assume that $\Delta = \mathcal O(n^{-\frac{1}{2}+\nu})$ as $n\to\infty$. In the end, we then have to substitute
\begin{align} \label{eq:defDelta}
\Delta = \frac{u\cdot z+z\cdot v}{\sqrt n}+\frac{u\cdot v}{n}. 
\end{align}

Having prepared all the ingredients necessary for the steepest descent analysis, we shall explicitly derive the asymptotic behavior of $I_{n,\tau}^d$ in the next section.

\section{Steepest descent analysis} \label{sec:steepest}

For $0<\tau<1$, our goal is to derive the asymptotic behavior of $I_{n,\tau}^d(\hat z_\pm + \sqrt{\hat z_\pm^2-2} \, \Delta_\pm)$ with a steepest descent analysis, under the assumption that $\Delta_\pm = \mathcal O(n^{-\frac{1}{2}+\nu})$ as $n\to\infty$. We deform the integration contour $\gamma_0$ to the circle $|s|=|a|^{-1} = e^{-\xi_+-\xi_-}$. We need to understand the behavior of $F(s)$ around the saddle point $s=a^{-1}$ and the pole (of the integrand) $s=\tau$. 

When $\tau=0$, the situation is comparable. Here we attempt to find the large $n$ behavior of $I_{n,0}^d(1+\Delta)$, under the condition that $\Delta=\mathcal O(n^{-\frac{1}{2}+\nu})$ as $n\to\infty$. The saddle point will coalesce with the pole $s=1$ in the limit $n\to\infty$. \\

\subsection{A saddle point coalescing with a pole} \label{sec:coalesc}

Under proper conditions, it is possible to perform a steepest descent analysis where a saddle point and a pole coalesce. 
Although our arguments can be extended to more general situations, we consider only the quadratic case on the real line. After the right preparations, Proposition \ref{prop:saddlePoleSteepest} turns out to be enough for our purposes. 

\begin{proposition} \label{prop:saddlePoleSteepest}
Let $\ell_1<0<\ell_2$. Uniformly for $p\in\mathbb C$ in compact sets 
with $\ell_1 <\Re p<\ell_2$, we have that 
\begin{align*}
\int_{\ell_1}^{\ell_2} e^{- n t^2} \frac{dt}{t-p} = -\pi i e^{- n p^2} \operatorname{erfc}(i \sqrt n p) + \mathcal O\left(\frac{e^{-n \ell_1^2} + e^{-n \ell_2^2}}{n}\right),
\end{align*}
as $n\to\infty$, where the path from $\ell_1$ to $\ell_2$ is such that $p$ is to the right of the path. 
\end{proposition}

\begin{proof}
The conditions imply that $\ell_1+\delta <\Re p<\ell_2-\delta$ for some constant $\delta>0$. 
Without loss of generality, we assume that $\Im p<0$ and that the curve from $\ell_1$ to $\ell_2$ is a line segment. First we note that
\begin{align*}
\left|\int_{-\infty}^{\ell_1} e^{- n t^2} \frac{dt}{t-p} + \int_{\ell_2}^\infty e^{- n t^2} \frac{dt}{t-p}\right|
&\leq \frac{1}{\Re p - \ell_1} \sqrt{\frac{\pi}{n}} \operatorname{erfc}\left(\sqrt n \ell_1\right) 
+ \frac{1}{\ell_2-\Re p} \sqrt{\frac{\pi}{n}} \operatorname{erfc}\left(\sqrt n \ell_2 \right)\\
&\leq \frac{1}{\delta} \frac{e^{-n \ell_1^2}}{\pi n \ell_1}
+\frac{1}{\delta} \frac{e^{-n \ell_2^2}}{\pi n \ell_2},
\end{align*}
as $n\to\infty$. Hence, we may just as well replace $\ell_1$ by $-T$ and $\ell_2$ by $T$ for some large $T>0$, and integrate over the line segment $[-T, T]$ in what follows. We create a closed integration contour by adding the line segments $[\Im p -T, p -\varepsilon]$, $[p+\varepsilon, \Im p+T]$ and $\pm T + i [\Im p, 0]$, and the upper semicircle from $p+\varepsilon$ to $p-\varepsilon$ for some small $\varepsilon>0$ (see Figure \ref{Fig1}). Performing the corresponding contour integration, and letting $T\to\infty$ and $\varepsilon\to 0$, we infer that
\begin{align*}
\int_{-\infty}^\infty e^{- n t^2} \frac{dt}{t-p}  = \pi i e^{- n p^2} - \operatorname{p. v. } \int_{-\infty}^\infty e^{- n (t+p)^2} \frac{dt}{t}.
\end{align*}
\begin{figure}[t]
\begin{center}
\begin{tikzpicture}[>=latex]
	\draw[-] (-7,0)--(7,0);
	\draw[-] (0,-4)--(-0,2);


	\node[above] at (0.2,0) {$0$};	
	\node[above] at (4.2,0) {$T$};	
	\node[above] at (-4.2,0) {$-T$};	
	\node[below] at (4.2,-2.1) {$T+i\Im p$};	
	\node[below] at (-4.2,-2.1) {$-T+i\Im p$};
	\node[below] at (1.05,-2.15) {$p$};	
	\node[below] at (0.4,-2.1) {$p-\varepsilon$};	
	\node[below] at (1.7,-2.1) {$p+\varepsilon$};	
	
	\draw[->, ultra thick] (-2,0) -- (-1.9,0);
	\draw[-, ultra thick] (-4,-2) -- (-4,0);
	\draw[-, ultra thick] (4,-2) -- (4,0);
	\draw[-, ultra thick] (-4,0) -- (4,0);
	\draw[-, ultra thick] (-4,-2) -- (0.5,-2);
	\draw[-, ultra thick] (1.5,-2) -- (4,-2);
	
	\draw[fill] (0,0) circle (0.08cm);
	\draw[fill] (-4,0) circle (0.08cm);
	\draw[fill] (4,0) circle (0.08cm);
	\draw[fill] (-4,-2) circle (0.08cm);
	\draw[fill] (4,-2) circle (0.08cm);
	\draw[fill] (1,-2) circle (0.08cm);
	\draw[fill] (0.5,-2) circle (0.08cm);
	\draw[fill] (1.5,-2) circle (0.08cm);
	
	\draw[-, ultra thick] (1.5,-2) arc (0:180:0.5);

%
%
	

	
\end{tikzpicture}
\caption{The integration contour used in the proof of Proposition \ref{prop:saddlePoleSteepest}. \label{Fig1}}
\end{center}
\end{figure}
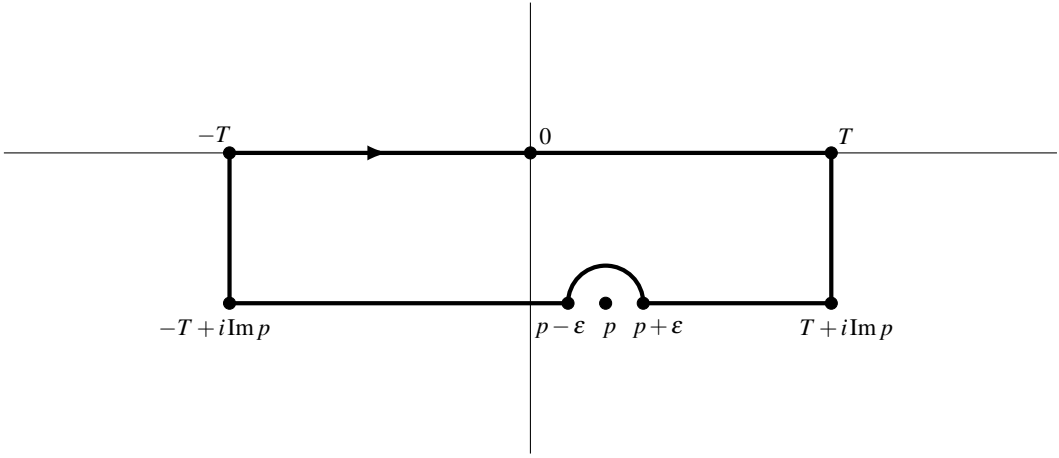
We can rewrite the principle value integral as
\begin{align*}
\operatorname{p. v. } \int_{-\infty}^\infty e^{- n (t+p)^2} \frac{dt}{t}
&= e^{-n p^2} \int_{-\infty}^\infty e^{- n t^2} \sum_{k=0}^\infty \frac{(-2 n p t)^{2k+1}}{(2 k+1)!} \frac{dt}{t}\\
&= \frac{e^{-n p^2}}{\sqrt n} \sum_{k=1}^\infty \frac{(-2 \sqrt n p)^{2 k+1}}{(2k+1)!} \Gamma\left(k+1/2\right)\\
&= 2 e^{- n p^2} \sqrt \pi \sum_{k=0}^\infty \frac{(-\sqrt{n} p)^{2k+1}}{k! (2k+1)} 
= -\pi i e^{- n p^2} \operatorname{erf}(i \sqrt{n} p)
\end{align*}
where we used the Legendre duplication formula $\Gamma(s)\Gamma(s+1/2) = 2^{1-2s} \sqrt\pi \, \Gamma(2s)$
for $s=k+1/2$, and the Taylor series expansion for the error function. Together with the half residue from before, this yields the complementary error function.
\end{proof}

\subsection{Conformal map}

Let us consider $\tau=0$. By \eqref{eq:F''ins0} we know that $F''(s_0)$ is approximately equal to $1$. Hence the steepest descent path going through $s_0$ is almost a a vertical line segment in a neighborhood of $s=s_0$ (and $s=1$). There is a conformal map transforming this steepest descent path to a real line segment, and this, after some rewriting, will allow the application of Proposition \ref{prop:saddlePoleSteepest}. Then by a simple Taylor series expansion of $F(s) = (1+\Delta) s - \log s$, we have
\begin{align*}
F(s) = F(s_0) + \sum_{k=2}^\infty \frac{(-1)^k}{k} \left(\frac{s-s_0}{s_0}\right)^k.
\end{align*}
The convergence radius of this series is given by $|s_0| = 1 + \mathcal O(n^{-\frac{1}{2}+\nu})$. For $n$ large enough, we may assume that the convergence radius $\rho$ is at least, say, $\frac{1}{2}$. Clearly then, we can find a conformal map $\phi : D(s_0; \rho) \to \mathbb C$ such that
\begin{align*}
F(s) - F(s_0) = - \phi(s)^2.
\end{align*}
Here $D(s_0; \rho)$ denotes the disc with radius $\rho$, centered at $s_0$. 
There are two such maps, but let us take the one with expansion
\begin{align} \label{eq:defConformalMap0}
\phi(s) = -\frac{i}{\sqrt 2} \frac{s-s_0}{s_0} + \frac{i}{3\sqrt 2} \left(\frac{s-s_0}{s_0}\right)^2 + \mathcal O((s-s_0)^3). 
\end{align}
Since the saddle point converges to $1$, we want to find out how $\phi$ behaves there. To this end, we have the following lemma. 

\begin{lemma}[$\tau=0$] \label{lem:asympPhi0}
Let $\phi$ be as defined above. Then we have as $n\to\infty$ that
\begin{align} \label{eq:lem:asympPhi0}
i \phi(1) = \frac{\Delta}{\sqrt 2} - \frac{\Delta^2}{3 \sqrt 2} + \mathcal O(n^{-\frac{3}{2}+3\nu}).
\end{align}
\end{lemma}

\begin{proof}
From \eqref{eq:defConformalMap} and \eqref{eq:defSaddlePoint0n} we have
\begin{align*}
i \phi(1) = \frac{1}{\sqrt 2} (s_0^{-1}-1) - \frac{1}{3\sqrt 2} (s_0^{-1}-1)^2+\mathcal O((s_0^{-1}-1)^3)
= \frac{\Delta}{\sqrt 2} - \frac{\Delta^2}{3\sqrt 2} +\mathcal O(n^{-\frac{3}{2}+3\nu}),
\end{align*}
as $n\to\infty$, where we used that $s_0^{-1} = 1+\Delta$.  
\end{proof}

Let us move to the case $0<\tau<1$, and derive an analogous result. 
As $n$ becomes large, by \eqref{eq:F''a-1Explicit}, we have that
\begin{align*}
F''(a^{-1}) = 2 \tau^{-2} \frac{|\sinh(\xi_\tau+i\eta)|^2}{\sinh(2\xi_\tau)} + \mathcal O(n^{-\frac{1}{2}+\nu})
\end{align*}
as $n\to\infty$. This is close to being a positive real number, hence the steepest descent direction is almost a vertical line in a neighborhood of $s=a^{-1}$ (and $s=\tau$). Again, there is a conformal map transforming this steepest descent path to a real line segment. Explicitly, we can write
\begin{align*}
F(s) - F(a^{-1}) &= \sum_{k=2}^\infty \left(\frac{(z_++z_-)^2}{2} \frac{(-1)^k}{(1+a^{-1})^{k+1}}-\frac{(z_+-z_-)^2}{2} \frac{1}{(1-a^{-1})^{k+1}}+\frac{(-a)^k}{k}\right) (s-a^{-1})^k
\end{align*}
The convergence radius of this series is given by $\min(\tau, 1-\tau)+\mathcal O(n^{-\frac{1}{2}+\nu})$. In particular, for $n$ large enough, we may assume that it has convergence radius at least $\rho := \frac{1}{2} \min(\tau, 1-\tau)$, uniformly for $\hat z \in \partial \mathscr E_\tau$. Thus, for $n$ large enough, we have
\begin{align*}
F(s) - F(a^{-1}) &= - \phi(s)^2,
\end{align*}
for some conformal map $\phi:D(a^{-1}, \rho)\to \mathbb C$, that depends on $n$. There are two such conformal maps, but we choose the one with expansion
\begin{align} \label{eq:defConformalMap}
\phi(s) = -i \sqrt{\frac{1}{2} F''(a^{-1})} (s-a^{-1}) - \frac{i}{6\sqrt 2} \frac{F'''(a^{-1})}{\sqrt{F''(a^{-1})}} (s-a^{-1})^2+ \mathcal O((s-a^{-1})^3).
\end{align}
By possibly taking $\rho$ smaller, the reader may convince oneself that the length of the real line segment, corresponding to the inverse image of the steepest descent contour, can be arranged to be bounded, both from below and from above, by a positive constant, uniformly for $n$ and $\hat z\in\partial\mathscr E_\tau$. In particular, we will consider a single line segment $[-r,r]$, with $r>0$, for all $\hat z\in\partial\mathscr E_\tau$ and all $n$, for which $\phi^{-1}([-r,r]) \subset D(a^{-1}, \rho)$ lies on the steepest descent contour through $s=a^{-1}$. 

\begin{lemma}[$0<\tau<1$] \label{lem:asympPhi}
We have
\begin{align} \label{eq:asympPhitaup}
i\, \phi(\tau)
= \frac{|\sinh(\xi_\tau+i\eta)|}{\sqrt{\sinh 2\xi_\tau}} (\Delta_++\Delta_-)
- \frac{\sqrt{\sinh 2\xi_\tau}}{|\sinh(\xi_\tau+i\eta)|} \frac{\Delta_+^2-\Delta_+\Delta_-+\Delta_-^2}{6}
+ \mathcal O(n^{-\frac{3}{2}+3\nu}),
\end{align}
as $n\to\infty$, uniformly for $\hat z \in \partial \mathscr E_\tau$.
\end{lemma}

\begin{proof}
Using \eqref{eq:F''a-1Explicit} and \eqref{eq:atoHighAccuracy}, we infer that
\begin{multline} \label{eq:F''a-1Explicit2}
\frac{\sqrt{\sinh 2\xi_\tau}}{|\sinh(\xi_\tau+i\eta)|} \sqrt{\frac{1}{2} F''(a^{-1})} (\tau-a^{-1}) =
\Delta_++\Delta_-+\frac{(\Delta_++\Delta_-)^2}{2}
- \coth(\xi_\tau+i\eta) \frac{\Delta_+^2}{2}
- \coth(\xi_\tau-i\eta) \frac{\Delta_-^2}{2}\\
+\frac{1}{2} (\Delta_++\Delta_-)
\left(\coth(\xi_\tau+i\eta) \Delta_+
 + \coth(\xi_\tau+i\eta) \Delta_-
 - \coth(2\xi_\tau) (\Delta_++\Delta_-)\right)
 +\mathcal O(n^{-\frac{3}{2}+3\nu})\\
 = \Delta_++\Delta_--\frac{\tau^2}{1-\tau^2} (\Delta_++\Delta_-)^2
 + \frac{\sinh(2\xi_\tau)}{|\sinh(\xi_\tau+i\eta)|^2} \frac{\Delta_+\Delta_-}{2}
 +\mathcal O(n^{-\frac{3}{2}+3\nu}).
\end{multline}
To arrive at the last line, we used the identities
\begin{align*}
\coth(2\xi_\tau) - 1 &= \frac{2\tau^2}{1-\tau^2},\\
\coth(\xi_\tau+i\eta)+\coth(\xi_\tau-i\eta) &= \frac{\sinh(2\xi_\tau)}{|\sinh(\xi_\tau+i\eta)|^2}.
\end{align*}
Also, we have
\begin{align} \label{eq:F''a-1Explicit3}
\frac{i}{6\sqrt 2}\frac{F'''(a^{-1})}{\sqrt{F''(a^{-1})}} (\tau-a^{-1})^2
= i\tau \frac{\sqrt{\sinh 2\xi_\tau}}{|\sinh(\xi_\tau+i\eta)|} \left(\frac{(\Re \hat z_+)^2}{(1+\tau)^4}+\frac{(\Im \hat z_+)^2}{(1-\tau)^4} - \frac{1}{6\tau^3}\right) \tau^2 (\Delta_++\Delta_-)^2
+\mathcal O(n^{-\frac{3}{2}+3\nu}).
\end{align}
This we can simplify by noticing that
\begin{align} \label{eq:F''a-1Explicit4}
\frac{(\Re \hat z_+)^2}{(1+\tau)^4}+\frac{(\Im \hat z_+)^2}{(1-\tau)^4}
= \frac{\cos^2 \eta}{2\tau (1+\tau)^2}+\frac{\sin^2 \eta}{2\tau (1-\tau)^2}
= \frac{1+\tau^2 - 2\tau \cos 2\eta}{2\tau (1-\tau^2)^2}
= \frac{|\sinh(\xi_\tau+i\eta)|^2}{\tau \sinh(2\xi_\tau) (1-\tau^2)}.
\end{align}
Combining \eqref{eq:F''a-1Explicit2} with \eqref{eq:F''a-1Explicit3} and \eqref{eq:F''a-1Explicit4}, and the defining relation \eqref{eq:defConformalMap} for $\phi$, we conclude that
\begin{multline*}
i\phi(\tau) = \frac{|\sinh(\xi_\tau+i\eta)|}{\sqrt{\sinh 2\xi_\tau}} (\Delta_++\Delta_-)
+ \frac{1}{2}\frac{\sqrt{\sinh 2\xi_\tau}}{|\sinh(\xi_\tau+i\eta)|} \Delta_+\Delta_-\\
+ \left(-\frac{|\sinh(\xi_\tau+i\eta)|}{\sqrt{\sinh 2\xi_\tau}} \frac{\tau^2}{1-\tau^2}+\frac{\sqrt{\sinh 2\xi_\tau}}{|\sinh(\xi_\tau+i\eta)|} 
\left(\frac{\tau^2|\sinh(\xi_\tau+i\eta)|^2}{\sinh(2\xi_\tau) (1-\tau^2)} - \frac{1}{6}\right)\right) (\Delta_++\Delta_-)^2
+\mathcal O(n^{-\frac{3}{2}+3\nu})\\
=  \frac{|\sinh(\xi_\tau+i\eta)|}{\sqrt{\sinh 2\xi_\tau}} (\Delta_++\Delta_-)
+ \frac{1}{2}\frac{\sqrt{\sinh 2\xi_\tau}}{|\sinh(\xi_\tau+i\eta)|} \Delta_+\Delta_-
- \frac{1}{6} \frac{\sqrt{\sinh 2\xi_\tau}}{|\sinh(\xi_\tau+i\eta)|} (\Delta_+^2+2\Delta_+\Delta_-+\Delta_-^2)
+\mathcal O(n^{-\frac{3}{2}+3\nu}),
\end{multline*}
as $n\to\infty$, uniformly for $\hat z\in\partial \mathscr E_\tau$, and this yields the result. 
\end{proof}

\subsection{Asymptotic behavior of $I_{n,\tau}^d$}

In this section, we perform a steepest descent analysis to obtain the asymptotic behavior of $I_{n,\tau}^d$, both for $0<\tau<1$ and $\tau=0$. Let us start with the latter, $\tau=0$. Here, we simply deform $\gamma_0$ to the steepest descent path through $s=s_0$. This is allowed because the steepest descent path ends in $-\infty$ as $n\to\infty$, from both directions (to see this, one should solve the contour lines $\Im(s_0^{-1} s - \log s) = \arg s_0$).  We may then write
\begin{align*}
I_{n,0}^d(z\cdot w) &= -\frac{1}{2\pi i} \oint_{\gamma_0} e^{n F(s)} 
\frac{1}{s-1}\left(1-\frac{\phi'(s)(s-1)}{\phi(s)-\phi(1)}\right) ds
-\frac{e^{n F(s_0)}}{2\pi i} \oint_{\gamma_0} \frac{e^{-n \phi(s)^2}}{s-1} \frac{\phi'(s)(s-1)}{\phi(s)-\phi(1)} ds\\
&= -\frac{1}{2\pi i} \oint_{\gamma_0} e^{n F(s)} 
\frac{1}{s-1}\left(1-\frac{\phi'(s)(s-1)}{\phi(s)-\phi(1)}\right) ds
-\frac{1}{2\pi i} e^{n F(s_0)} \int_{-\infty}^\infty \frac{e^{-n s^2}}{s-\phi(1)} ds.
\end{align*}
The first integral can be done with a standard steepest descent procedure, where we rescale variables locally around $s=s_0$ (with a factor $n^{-\frac{1}{2}+\nu}$). The reader may verify that the assumption that $0<\nu<\frac{1}{6}$ is necessary for this procedure to work (the cubic term has to be small). We see that
\begin{align*}
\frac{1}{s-1}\left(1-\frac{\phi'(s)(s-1)}{\phi(s)-\phi(1)}\right)
= \frac{\phi(s)-\phi(1)+\phi'(s)(1-s)}{(s-1)(\phi(s)-\phi(1))}
= - \frac{1}{2} \frac{\phi''(1)}{\phi'(1)}+\mathcal O(s-1)
= \frac{1}{3} + \mathcal O(\Delta) +\mathcal O(s-1).
\end{align*}
Performing the steepest descent analysis and invoking Proposition \ref{prop:saddlePoleSteepest}, we get
\begin{align} \label{eq:prop:asympIntaudX0pre}
I_{n,0}^d(z\cdot w)
= \frac{1}{2} e^{n F(1)} \operatorname{erfc}(i\sqrt n\phi(1)) 
+ \left(-\frac{1}{3}+\mathcal O(n^{-1+2\nu})\right) \frac{e^{n F(s_0)}}{\sqrt{2\pi F''(s_0) n}}.
\end{align}
We thus get the following proposition.

\begin{proposition}[$\tau=0$] \label{prop:asympIntaudX0}
We have as $n\to\infty$ that
\begin{align} \label{eq:prop:asympIntaudX0}
e^{-n F(1)} I_{n,0}^d(1+\Delta)
= \frac{1}{2} \operatorname{erfc}\left(\sqrt n\frac{\Delta}{\sqrt 2}\right)
+ \frac{1}{\sqrt{2\pi n}} \exp\left(-n\frac{\Delta^2}{2}\right)
\left(\frac{n \Delta^2 - 1}{3}+\mathcal O(n^{-\frac{1}{2}+3\nu})\right).
\end{align}
\end{proposition}

\begin{proof}
By the preceeding, we have
\begin{align} \label{eq:asympInttau0F1}
e^{-n F(1)} I_{n,0}^d(1+\Delta)
= \frac{1}{2} \operatorname{erfc}(i\sqrt n\phi(1)) 
+ \left(-\frac{1}{3}+\mathcal O(n^{-1+2\nu})\right) \frac{e^{-n \phi(1)^2}}{\sqrt{2\pi F''(s_0) n}} .
\end{align}
By Taylor expanding and Proposition \ref{lem:asympPhi0}, we have
\begin{align} \label{eq:asympInttau0F1b}
\operatorname{erfc}(i\sqrt n \phi(1)) &= \operatorname{erfc}\left(\sqrt n\frac{\Delta}{\sqrt 2}\right)
- \frac{2}{\sqrt \pi} \exp\left(- n \frac{\Delta^2}{2}\right) \left(-\sqrt n\frac{\Delta^2}{3\sqrt 2}+\mathcal O(n^{-1+3\nu})\right),
\end{align}
and similarly
\begin{align} \label{eq:asympInttau0F1c}
e^{-n \phi(1)^2} &= \exp\left(-n\frac{\Delta^2}{2}\right) (1+\mathcal O(n^{-\frac{1}{2}+3\nu})). 
\end{align}
Plugging \eqref{eq:asympInttau0F1b} and \eqref{eq:asympInttau0F1c} in \eqref{eq:asympInttau0F1}, and using that $F''(s_0) = s_0^2$ and $s_0^{-1} = 1+\Delta$, we arrive at the result. 
\end{proof}
Now we turn to our integral \eqref{eq:defIntaud} for the case $0<\tau<1$, where we deform $\gamma_0$ to the circle $|s|=|a|^{-1}$. 
As the reader may verify, similar to what we did for $\tau=0$, we can rewrite our integral as
\begin{multline} \label{eq:InRewrite3Terms}
I_{n,\tau}^d(z_\pm) = 
- \frac{1}{2\pi i} \oint_{\gamma_0} \frac{e^{n F(s)}}{s-\tau} \left(\frac{1}{(1-s^2)^\frac{d}{2}}-\frac{1}{(1-\tau^2)^\frac{d}{2}}\right) ds
- \frac{1}{2\pi i} \frac{1}{(1-\tau^2)^\frac{d}{2}} \int_{\phi(-r)}^{\phi(r)} \frac{e^{n F(s)}}{s-\tau} \left(1 - \frac{\phi'(s) (s-\tau)}{\phi(s)-\phi(\tau)}\right) ds\\
-\frac{1}{2\pi i} \frac{e^{n F(a^{-1})}}{(1-\tau^2)^\frac{d}{2}}\int_{-r}^r \frac{e^{-n s^2}}{s-\phi(\tau)} ds + \mathcal O(e^{n F(a^{-1})} e^{-c n})
\end{multline}
uniformly for $\hat z\in\partial\mathscr E_\tau$ as $n\to\infty$, for some constant $c>0$. 
The first two integrals on the RHS of \eqref{eq:InRewrite3Terms} can be asymptotically solved with a standard steepest descent procedure, where we rescale variables locally around $s=a^{-1}$ (with a factor $n^{-\frac{1}{2}+\nu}$). Again, the reader may verify that the assumption that $0<\nu<\frac{1}{6}$ is necessary for this procedure to work. We note that
\begin{align*}
\frac{1}{s-\tau} \left(1 - \frac{\phi'(s) (s-\tau)}{\phi(s)-\phi(\tau)}\right)
= \frac{\phi(s)-\phi(\tau) + \phi'(s)(\tau-s)}{(s-\tau) (\phi(s)-\phi(\tau))}
= -\frac{1}{2} \frac{\phi''(\tau)}{\phi'(\tau)} + \mathcal O(s-\tau). 
\end{align*}
Hence, executing the steepest descent method for the first two integrals, and applying Proposition \ref{prop:saddlePoleSteepest}, we get
\begin{align*}
I_{n,\tau}^d(z_\pm) = \frac{1}{2} \frac{e^{n F(\tau)}}{(1-\tau^2)^\frac{d}{2}} \operatorname{erfc}(i\sqrt n\phi(\tau))
- \frac{1}{\sqrt{\pi n}} \frac{\tau \sqrt{\sinh 2\xi_\tau}}{|\sinh(\xi_\tau+i\eta)|} \frac{e^{n F(a^{-1})}}{(1-\tau^2)^\frac{d}{2}}
\left(\frac{d \tau}{1-\tau^2} - \frac{1}{2}\frac{\phi''(\tau)}{\phi'(\tau)}+\mathcal O(n^{-\frac{1}{2}+3\nu})\right).
\end{align*}
We used here that it did not matter whether $\phi(\tau)$ was to the right or the left of the contour, by application of the residue theorem, since we started with a closed contour. 
In particular, we have
\begin{align*}
(1-\tau^2)^\frac{d}{2} e^{-n F(\tau)} I_{n,\tau}^d(z_\pm) = \frac{1}{2} \operatorname{erfc}(i\sqrt n\phi(\tau))
- \frac{1}{\sqrt{\pi n}} \frac{\tau \sqrt{\sinh 2\xi_\tau}}{|\sinh(\xi_\tau+i\eta)|} e^{n \phi(\tau)^2}
\left(\frac{d \tau}{1-\tau^2} - \frac{1}{2}\frac{\phi''(\tau)}{\phi'(\tau)}+\mathcal O(n^{-\frac{1}{2}+3\nu})\right).
\end{align*}
as $n\to\infty$, uniformly for $\hat z \in \partial \mathscr E_\tau$. 

Summarizing, we have the following theorem.

\begin{proposition}[$0<\tau<1$] \label{prop:asympIntaudX}
Let $\hat z \in \partial \mathscr E_\tau$. Let $p$ be as in Lemma \ref{lem:asympPhi}. Then we have
\begin{multline} \label{eq:prop:asympIntaudX}
(1-\tau^2)^\frac{d}{2} e^{-n F(\tau)} I_{n,\tau}^d(z_\pm) 
= \frac{1}{2} \operatorname{erfc}\left(\frac{|\sinh(\xi_\tau+i\eta)|}{\sqrt{\sinh 2\xi_\tau}} \sqrt n (\Delta_++\Delta_-)\right)\\
+ \frac{1}{\sqrt{\pi n}} \frac{\sqrt{\sinh 2\xi_\tau}}{|\sinh(\xi_\tau+i\eta)|} \exp\left(-\frac{|\sinh(\xi_\tau+i\eta)|^2}{\sinh 2\xi_\tau} n(\Delta_++\Delta_-)^2\right)\\
\times \left(n\frac{\Delta_+^2-\Delta_+\Delta_-+\Delta_-^2}{3}  
- \frac{1}{6} \frac{\sinh 2\xi_\tau}{|\sinh(\xi_\tau+i\eta)|^2}
- \tau^2\frac{d-1}{1-\tau^2} + \mathcal O(n^{-\frac{1}{2}+3\nu})\right),
\end{multline}
as $n\to\infty$, uniformly for $\hat z \in \partial \mathscr E_\tau$. 
\end{proposition}

\begin{proof}
This follows from the formulae above, and
\begin{align*}
\operatorname{erfc}(i\sqrt n\phi(\tau))
&= \operatorname{erfc}\left(\frac{|\sinh(\xi_\tau+i\eta)|}{\sqrt{\sinh 2\xi_\tau}} \sqrt n (\Delta_++\Delta_-)\right)\\
&\qquad - \frac{2}{\sqrt \pi} \exp\left(-\frac{|\sinh(\xi_\tau+i\eta)|^2}{\sinh 2\xi_\tau} n (\Delta_++\Delta_-)^2\right) 
\left(- \frac{\sqrt{\sinh 2\xi_\tau}}{|\sinh(\xi_\tau+i\eta)|} \sqrt n \frac{\Delta_+^2-\Delta_+\Delta_-+\Delta_-^2}{6}
+ \mathcal O(n^{-1+3\nu})\right),\\
e^{n\phi(\tau)^2} 
&= \exp\left(-\frac{|\sinh(\xi_\tau+i\eta)|^2}{\sinh 2\xi_\tau} n (\Delta_++\Delta_-)^2\right) 
\left(1 + \mathcal O(n^{-\frac{1}{2}+3\nu})\right),
\end{align*}
which are obtained by taking a Taylor series and using Lemma \ref{lem:asympPhi}. Finally, with arguments similar to those in the proof of \text{Lemma \ref{lem:asympPhi}}, we have
\begin{align*}
\frac{d \tau}{1-\tau^2} - \frac{1}{2}\frac{\phi''(\tau)}{\phi'(\tau)}
= \frac{(d-1)\tau}{1-\tau^2} + \frac{1}{6\tau} \frac{\sinh 2\xi_\tau}{|\sinh(\xi_\tau+i\eta)|^2} 
+ \mathcal O(n^{-\frac{1}{2}+\nu}).
\end{align*}
\end{proof}

\section{Explicit form of relevant expressions} \label{sec:relevant}

To prove the main theorems we need to start plugging $u, v$ and $z$ into our formulae. Doing so, we can state Lemma \ref{lem:asympPhi0} and Lemma \ref{lem:asympPhi} differently. 

\begin{lemma}[$\tau=0$] \label{lem:asympPhiuvz0}
Fix $0<\nu<\frac{1}{6}$. For $z\in \mathbb C^d$ with $|z|=1$, denote by $\textbf{n}$ the outward unit normal vector on the $2d$-dimensional unit sphere $\partial E_0^d$ in $z$. We have
\begin{align} \label{eq:asympPhiOther0}
i\sqrt n \phi(1) = \frac{u \cdot \textbf{n}+\textbf{n}\cdot v}{\sqrt 2} + \frac{3 u\cdot v - (u\cdot \textbf{n}+\textbf{n}\cdot v)^2}{3\sqrt{2 n}}+\mathcal O(n^{-1+3\nu}),
\end{align}
as $n\to\infty$, uniformly for $|z|=1$ and $u, v\in\mathbb C^d$ with $u, v=\mathcal O(n^\nu)$.  
Furthermore, if $u = v = \lambda \textbf{n}$ with $\lambda\in\mathbb R$, then we have
\begin{align*}
i\sqrt n \phi(1) = \sqrt 2 \lambda -\frac{\lambda^2}{3 \sqrt{2 n}}+\mathcal O(n^{-1+3\nu}),
\end{align*}
as $n\to\infty$, uniformly for $|z|=1$ and $\lambda = \mathcal O(n^\nu)$. 
\end{lemma}

\begin{proof}
It is well-known that the unit normal vector on the boundary of the $2d$-dimensional unit ball in $z$ is given simply by $z$. Substituting \eqref{eq:defDelta} in \eqref{eq:lem:asympPhi0}, we find that
\begin{align*}
\sqrt{2 n} i \phi(1) &= u\cdot z+z\cdot v + \frac{u\cdot v}{\sqrt n} - \frac{(u\cdot z+z\cdot v)^2}{3} + \mathcal O(n^{-1+3\nu}).
\end{align*}
Identifying $\textbf{n}=z$ in this equation yields the first part of the lemma. The second part follows trivially from the first part. 
\end{proof}

With a little more effort, we find an analogous statement for the case $0<\tau<1$. 

\begin{lemma}[$0<\tau<1$] \label{lem:asympPhiuvz}
Fix $0<\nu<\frac{1}{6}$. For $z\in\partial \mathscr E_\tau^d$, denote by $\textbf{n}$ the outward unit normal vector on $\partial \mathscr E_\tau^d$ in $z$. We have
\begin{align} \label{eq:asympPhiOther}
i\sqrt n \phi(\tau) = \frac{u \cdot \textbf{n}+\textbf{n}\cdot v}{\sqrt 2} + \frac{\mathcal O(|u|^2+|v|^2)}{\sqrt n}+\mathcal O(n^{-1+3\nu}),
\end{align}
as $n\to\infty$, uniformly for $z\in \partial \mathscr E_\tau$ and $u, v\in\mathbb C^d$ with $u, v=\mathcal O(n^\nu)$.  
Furthermore, if $u = v = \lambda \textbf{n}$ with $\lambda\in\mathbb R$, then we have
\begin{align*}
i\sqrt n \phi(\tau) = \sqrt 2 \lambda + \frac{1}{6} \left(\frac{\sqrt{\sinh 2\xi_\tau}}{|\sinh(\xi_\tau+i\eta)|}\right)^3 \frac{\lambda^2}{\sqrt n}+\mathcal O(n^{-1+3\nu}),
\end{align*}
as $n\to\infty$, uniformly for $z\in \partial \mathscr E_\tau^d$ and $\lambda = \mathcal O(n^\nu)$. 
\end{lemma}

\begin{proof}
First, we need to describe the outward unit normal vector $\textbf{n}$. The hyperellipsoid $\partial \mathscr E_\tau^d$ can be written as a $(2d-1)$-dimensional hypersurface $f(z) = 0$, where
\begin{align*}
f(z) = \frac{1-\tau}{1+\tau} |\Re z|^2 + \frac{1+\tau}{1-\tau} |\Im z|^2 - 1.
\end{align*}
A normal vector is then given by the gradient of $f$, which, when normalized, gives
\begin{align} \label{eq:defOutwardNormal}
\textbf{n} = \frac{(1-\tau)^2 \Re z + i (1+\tau)^2 \Im z}{\sqrt{(1-\tau)^4 |\Re z|^2 + (1+\tau)^4 |\Im z|^2}}.
\end{align}
This is the outward normal vector on $\partial \mathscr E_\tau^d$. Now let us investigate $\Delta_++\Delta_-$, as defined via \eqref{eq:defDelta+-}. 
We notice that
\begin{align*}
\frac{1}{\sinh(\xi_\tau+i\eta)} + \frac{1}{\sinh(\xi_\tau-i\eta)} &= \frac{\sinh(\xi_\tau) \cos\eta}{|\sinh(\xi_\tau+i\eta|^2}
= \frac{1}{|\sinh(\xi_\tau+i\eta)|}\frac{(1-\tau)^2|\Re z|}{\sqrt{(1-\tau)^4 |\Re z|^2 + (1+\tau)^4 |\Im z|^2}},\\
\frac{1}{\sinh(\xi_\tau+i\eta)} - \frac{1}{\sinh(\xi_\tau-i\eta)} &= i \frac{\cosh(\xi_\tau) \sin\eta}{|\sinh(\xi_\tau+i\eta|^2}
=  \frac{1}{|\sinh(\xi_\tau+i\eta)|}\frac{(1+\tau)^2 i|\Im z|}{\sqrt{(1-\tau)^4 |\Re z|^2 + (1+\tau)^4 |\Im z|^2}}.
\end{align*}
Hence we have
\begin{multline} \label{eq:Delta++Delta-intermsofbfn}
\sqrt 2 \sqrt{(1-\tau)^4 |\Re z|^2 + (1+\tau)^4 |\Im z|^2} \frac{|\sinh(\xi_\tau+i\eta)|}{\sqrt{\sinh 2\xi_\tau}} (\Delta_++\Delta_-)\\
=
(1-\tau)^2|\Re z| \left(\sqrt{|\Re z|^2+\frac{u+\overline v}{\sqrt n}\cdot (\Re z)+\frac{(u+\overline v)^2}{4 n}} - |\Re(z)|\right)\\
+ (1+\tau)^2|\Im z| \left(\sqrt{|\Im z|^2 + \frac{u-\overline v}{\sqrt n}\cdot (i\Im z)-\frac{(u-\overline v)^2}{4 n}} - |\Im z|\right).
\end{multline}
This, for $|\Re z|, |\Im z|> \frac{|u|+|v|}{\sqrt n}$, we can rewrite as
\begin{align*}
& \frac{(1-\tau)^2}{2} \frac{u+\overline v}{\sqrt n} \cdot (\Re z) + \frac{(1+\tau)^2}{2}\frac{u-\overline v}{\sqrt n}\cdot (i\Im z)\\
&\quad + \frac{(u+\overline v)^2-((u+\overline v)\cdot \frac{\Re z}{|\Re z|})^2}{4 n} 
\frac{|\Re z|+\frac{1}{2} \frac{u+\overline v}{\sqrt n} \cdot (\Re z)}{|\Re z|+\frac{1}{2} \frac{u+\overline v}{\sqrt n} \cdot (\Re z)+\sqrt{|\Re z|^2+\frac{u+\overline v}{\sqrt n}\cdot (\Re z)+\frac{(u+\overline v)^2}{4 n}}}\\
&\quad  - \frac{(u-\overline v)^2-((u-\overline v)\cdot \Im z)^2}{4 n} 
\frac{|\Im z|+\frac{1}{2} \frac{u-\overline v}{\sqrt n} \cdot (i\Im z)}{|\Im z|+\frac{1}{2} \frac{u-\overline v}{\sqrt n} \cdot (i\Im z)+\sqrt{|\Im z|^2 + \frac{u-\overline v}{\sqrt n}\cdot (i\Im z)-\frac{(u-\overline v)^2}{4 n}}}.
\end{align*}
Since $\Re \sqrt\zeta\geq 0$ for all $\zeta\in\mathbb C$, we have for all $\zeta\in\mathbb C\cup\{\infty\}$ that
\begin{align*}
\left|\frac{1}{1+\sqrt \zeta}\right|\leq 1. 
\end{align*}
We conclude that, when $|\Re z|, |\Im z|> \frac{|u|+|v|}{\sqrt n}$, we have
\begin{align*}
\left|\sqrt 2 \frac{|\sinh(\xi_\tau+i\eta)|}{\sqrt{\sinh 2\xi_\tau}} (\Delta_++\Delta_-) - \frac{u\cdot \textbf{n}+\textbf{n}\cdot v}{2\sqrt n} \right|
\leq \frac{1}{2 n}\frac{|u+\overline v|^2+|u-\overline v|^2}{\sqrt{(1-\tau)^4 |\Re z|^2 + (1+\tau)^4 |\Im z|^2}}.
\end{align*}
The reader may check by explicit calculation that this estimate is still valid for $|\Re z|\leq\frac{|u|+|v|}{\sqrt n}$ or $|\Im z|\leq\frac{|u|+|v|}{\sqrt n}$. Plugging our estimate into \eqref{eq:asympPhitaup}, we arrive at the first statement of the lemma.

When $u=v=\lambda \textbf{n}$, equation \eqref{eq:Delta++Delta-intermsofbfn} simplifies considerably. We have 
\begin{align*}
(u+\overline v) \cdot (\Re z) &= 2\lambda (\Re \textbf{n}) \cdot (\Re z) 
= \frac{2\lambda(1-\tau)^2}{\sqrt{(1-\tau)^4 |\Re z|^2 + (1+\tau)^4 |\Im z|^2}} |\Re z|^2,\\
(u-\overline v) \cdot (i\Im z) &= 2\lambda (\Im \textbf{n}) \cdot (\Im z)
= \frac{2\lambda(1+\tau)^2}{\sqrt{(1-\tau)^4 |\Re z|^2 + (1+\tau)^4 |\Im z|^2}} |\Im z|^2,
\end{align*}
and
\begin{align*}
\frac{(u+\overline v)^2}{4} &= \lambda^2 |\Re \textbf{n}|^2 = \lambda^2 \frac{(1-\tau)^4}{(1-\tau)^4 |\Re z|^2 + (1+\tau)^4 |\Im z|^2} |\Re z|^2,\\
-\frac{(u-\overline v)^2}{4} &= \lambda^2 |\Im \textbf{n}|^2 = \lambda^2 \frac{(1+\tau)^4}{(1-\tau)^4 |\Re z|^2 + (1+\tau)^4 |\Im z|^2} |\Im z|^2.
\end{align*}
Hence we find that
\begin{multline*}
\sqrt 2 \frac{|\sinh(\xi_\tau+i\eta)|}{\sqrt{\sinh 2\xi_\tau}} (\Delta_++\Delta_-)\\
= \frac{(1-\tau)^2|\Re z|}{\sqrt{(1-\tau)^4 |\Re z|^2 + (1+\tau)^4 |\Im z|^2}} 
\left(|\Re z|\left(1+\frac{\lambda (1-\tau)^2}{\sqrt n\sqrt{(1-\tau)^4 |\Re z|^2 + (1+\tau)^4 |\Im z|^2}}\right) - |\Re(z)|\right)\\
+ \frac{(1+\tau)^2|\Im z|}{\sqrt{(1-\tau)^4 |\Re z|^2 + (1+\tau)^4 |\Im z|^2}} 
\left(|\Im z|\left(1+\frac{\lambda (1+\tau)^2}{\sqrt n\sqrt{(1-\tau)^4 |\Re z|^2 + (1+\tau)^4 |\Im z|^2}}\right) - |\Im z|\right) = \frac{\lambda}{\sqrt n}.
\end{multline*}
In particular, we have
\begin{align} \label{eq:isqrtnphitauu=v}
i\sqrt n \phi(\tau) = \sqrt 2 \, \lambda - \frac{\sqrt{\sinh 2\xi_\tau}}{|\sinh(\xi_\tau+i\eta)|} \frac{\Delta_+^2-\Delta_+\Delta_-+\Delta_-^2}{6 \sqrt n}
+ \mathcal O(n^{-1+3\nu}). 
\end{align}
Similar to before, we can show that
\begin{align*}
\sqrt n \Delta_\pm
= \frac{\sqrt{\sinh 2\xi_\tau}}{\sinh(\xi_\tau\pm i\eta)}
\frac{(1-\tau)^2 |\Re z| \pm i (1+\tau)^2 |\Im z|}{\sqrt{(1-\tau)^4 |\Re z|^2 + (1+\tau)^4 |\Im z|^2}} \lambda
= \frac{\sqrt{\sinh 2\xi_\tau}}{|\sinh(\xi_\tau+ i\eta)|} \lambda.
\end{align*}
Plugging these into \eqref{eq:isqrtnphitauu=v} finishes the proof. 
\end{proof}

\begin{lemma} \label{lem:Ftauuvz}
Let $\tau=0$. Then we have
\begin{align} \label{eq:lem:Ftauuvz1}
\exp\left(-\frac{|\sqrt n \, z+u|^2+|\sqrt n \, z+v|^2}{2}\right) e^{n F(1)} = \exp\left(u\cdot v-\frac{|u|^2+|v|^2}{2}\right). 
\end{align}
Now let $0<\tau<1$. There exist unimodular functions $c_n : \partial\mathscr E_\tau^d \times \mathbb C\to \mathbb T$ such that
\begin{align} \label{eq:lem:Ftauuvz2}
c_n(z, u) \overline{c_n(z, v)} \sqrt{\omega(\sqrt n \, z_+) \omega(\sqrt n \, z_-)} e^{n F(\tau)} &= \exp\left(u\cdot v-\frac{|u|^2+|v|^2}{2}\right). 
\end{align}
\end{lemma}

\begin{proof}
Let us start with the case $\tau=0$. We see that
\begin{align*}
|\sqrt n \, z+u|^2+|\sqrt n \, z+v|^2 = 2 n + |u|^2+|v|^2 + 2 \sqrt n \, (u+v) \cdot \overline z.
\end{align*}
It readily follows from the definition of $F$ for $\tau=0$ that
\begin{align*}
F(1) = \left(z+\frac{u}{\sqrt n}\right) \cdot \left(z+\frac{v}{\sqrt n}\right)
= 1 + \frac{u\cdot v}{n} + \frac{(u+v) \cdot \overline z}{\sqrt n}.
\end{align*}
Combining these two identities, we obtain \eqref{eq:lem:Ftauuvz1}. Let move on to the case $0<\tau<1$. 
We know that
\begin{multline} \label{eq:unimodFtauGinibre}
\sqrt{\omega(\sqrt n \, z_+) \omega(\sqrt n \, z_-))}
= \prod_{j=1}^d \sqrt{\omega(\sqrt n \, z_j+u_j) \omega(\sqrt n \, z_j+v_j)}\\
= \exp\left(-\frac{1}{2} \sum_{j=1}^d (|\sqrt n \, z_j+u_j|^2 - \frac{\tau}{2} ((\sqrt n \, z_j+u_j)^2+\overline{(\sqrt n \, z_j+u_j)}^2) 
+ |\sqrt n \, z_j+v_j|^2 - \frac{\tau}{2} ((\sqrt n \, z_j+v_j)^2+\overline{(\sqrt n \, z_j+v_j)}^2) )\right).
\end{multline}
We notice that
\begin{align} \label{eq:unimodFtauGinibre2}
\sum_{j=1}^d (|\sqrt n \, z_j+u_j|^2 - \frac{\tau}{2} ((\sqrt n \, z_j+u_j)^2+\overline{(\sqrt n \, z_j+u_j)}^2)
= n (|z|^2 - \tau \Re(z^2)) + |u|^2 - \tau \Re(u^2)
+ \sqrt n (u\cdot (\overline z - \tau z)+(z-\tau\overline z)\cdot u),
\end{align}
and, similarly with $u$ replaced by $v$. 
On the other hand, we have
\begin{align} \nonumber
F(\tau) &= \frac{\tau (z_++z_-)^2}{2(1+\tau)} - \frac{\tau (z_+-z_-)^2}{2(1-\tau)}\\ \nonumber
&= \frac{1-\tau}{4} \sum_{j=1}^d \left(2 \Re z_j + \frac{u_j+\overline v_j}{\sqrt n}\right)^2
- \frac{1+\tau}{4} \sum_{j=1}^d \left(2 i \Im z_j + \frac{u_j-\overline v_j}{\sqrt n}\right)^2\\ \label{eq:unimodFtauGinibre3}
&= |z|^2 - \tau \Re(z^2)
+ \frac{u\cdot v - \frac{\tau}{2} (u^2+v^2)}{n} 
+ \frac{u \cdot (\overline z -\tau z) + (z - \tau \overline z) \cdot v}{\sqrt n}.
\end{align}
Plugging the identities \eqref{eq:unimodFtauGinibre2} (also for $u$ replaced by $v$) and \eqref{eq:unimodFtauGinibre3} into \eqref{eq:unimodFtauGinibre}, we find that
\begin{align*}
\sqrt{\omega(\sqrt n \, z_+) \omega(\sqrt n \, z_-)} e^{n F(\tau)}
= \exp\left(u\cdot v-\frac{|u|^2+|v|^2}{2}\right)
\exp\left(- i \tau \Im(u^2-v^2) \right) 
\exp\left(i \tau \sqrt n \, \Im((z-\tau \overline z) \cdot (u-v)) \right).
\end{align*}
Clearly then, to get \eqref{eq:lem:Ftauuvz2}, we should define the unimodular functions by 
\begin{align*}
c_n(z, u) = \exp\left(i \tau \Im(u^2) \right) 
\exp\left(-i \tau \sqrt n \, \Im((z-\tau \overline z) \cdot u) \right).
\end{align*}
\end{proof}

The following is more or less a direct consequence of Proposition \ref{prop:KnisKnford=1} and Lemma \ref{lem:Ftauuvz}, we omit a proof. The result means that we find the higher dimensional analogue of the Ginibre kernel, i.e., a factorisation in $d$ Ginibre kernels, as a scaling limit in the bulk. In \cite{ADM}, this was proved for $0<\tau<1$. 

\begin{corollary}[$\tau=0$] \label{cor:bulkScalingLimit}
Take $|z|<1$ and $u, v\in\mathbb C^d$. Then we have
\begin{align*}
\lim_{n\to\infty} \mathscr K_n\left(\sqrt n \, z+u, \sqrt n \, z+v\right)
= \frac{1}{\pi^d} \exp\left(u\cdot v - \frac{|u|^2+|v|^2}{2}\right).
\end{align*}
The convergence is uniform on compact sets of $|z|<1$ and $u, v$. 
\end{corollary}

\section{Proof of the main theorems} \label{sec:proofs}

We now have all the ingredients necessary to prove the main results. 

\begin{proof}[Proof of Theorem \ref{thm:mainThmFaddeevaKernel}]
We first treat the case $\tau=0$. Substituting \eqref{eq:asympPhiOther0} in \eqref{eq:prop:asympIntaudX0pre}, and using Taylor expansions in the same way as in the proof of Proposition \ref{prop:asympIntaudX0}, we find that
\begin{multline} \label{eq:proofthm:mainThmFaddeevaKernel}
e^{-n F(1)} I_{n,0}^d\left(\left(z+\frac{u}{\sqrt n}\right)\cdot\left(z+\frac{v}{\sqrt n}\right)\right)
= \frac{1}{2} \operatorname{erfc}\left(\frac{u\cdot \textbf{n}+\textbf{n}\cdot v}{\sqrt 2}\right)\\
+ \frac{1}{\sqrt{2\pi n}}\exp\left(\frac{(u\cdot \textbf{n}+\textbf{n}\cdot v)^2}{2}\right)
\left(\frac{(u\cdot \textbf{n}+\textbf{n}\cdot v)^2 - 3 u\cdot v-1}{3} + \mathcal O(n^{-\frac{1}{2}+3\nu})\right),
\end{multline}
as $n\to\infty$, uniformly for $|z|=1$ and $u, v\in\mathbb C^d$ such that $u, v = \mathcal O(n^\nu)$. Multiplying by the remaining factors in \eqref{eq:defKernelGinUEd}, and using Lemma \ref{lem:Ftauuvz}, we obtain the statement of Theorem \ref{thm:mainThmFaddeevaKernel} for $\tau=0$. 
Next, we treat the case $0<\tau<1$. 
Starting from the formula \eqref{eq:defKnInIn}, we simply combine Lemma \ref{lem:asympPhiuvz} and Lemma \ref{lem:Ftauuvz} with Proposition \ref{prop:asympIntaudX} (which is valid with $\Delta_+^2-\Delta_+\Delta_-+\Delta_-^2$ replaced by $\mathcal O(|u|^2+|v|^2)$). 
\end{proof}

As the proof shows, we can be more precise about the error when $\tau=0$. In the case of the elliptic Ginibre ensemble, i.e., the case $d=1$ and $0<\tau<1$, we can also be more precise about the error. The folowing result (but slightly weaker) was proved by Byun and Ebke in \cite{ByEb}. We consider our novel derivation of this result of independent interest. 

\begin{proposition} \label{prop:d=1FaddeevaSlightlyStronger}
Pick $0<\nu<\frac{1}{6}$. Let $\mathscr K_n$ be the kernel of the eigenvalues of the elliptic Ginibre ensemble with parameter $0<\tau<1$, as defined in \eqref{eq:defKerneleGinUE}. Let $z\in \partial \mathscr E_\tau$, and denote by $\textbf{n}$ and $\kappa$, respectively, the outward unit normal vector and the curvature, of $\partial \mathscr E_\tau$ in $z$. Then there exist continuous unimodular functions $c_n : \partial \mathscr E_\tau \times \mathbb C\to \mathbb T$ such that 
\begin{multline} \label{eq:FaddeevaRiserByunEbke2}
c_n(z, u) \overline{c_n(z, v)} \mathscr K_n(\sqrt n \, z+u \, \textbf{n}, \sqrt n \, z+v \, \textbf{n})
= \frac{1}{2 \pi} \exp\left(u \overline v-\frac{|u|^2+|v|^2}{2}\right) \operatorname{erfc}\left(\frac{u + \overline v}{\sqrt 2}\right) \\
 +  \frac{\kappa}{\sqrt n} \exp\left(- \frac{|u|^2+u^2+|v|^2+\overline{v}^2}{2}\right) 
 \left(\frac{u^2+\overline v^2 - u\overline v - 1}{3\sqrt{2 \pi^3}} 
   + \mathcal O(n^{-\frac{1}{2}+3\nu})\right),
\end{multline}
as $n\to\infty$, uniformly for $z\in \partial\mathscr E_\tau$, and $u, v\in\mathbb C$ such that $u, v = \mathcal O(n^\nu)$.
\end{proposition}

\begin{proof}
In this case we have $z_+ = \sqrt{\sinh(2\xi_\tau)} (z+\frac{u}{\sqrt n} \textbf{n})$ and $z_- = \sqrt{\sinh(2\xi_\tau)} \, \overline{(z+\frac{v}{\sqrt n} \textbf{n})}$. The outward unit normal vector is now given simply by
\begin{align*}
\textbf{n} = \frac{\sinh(\xi_\tau+i\eta)}{|\sinh(\xi_\tau+i\eta)|}.
\end{align*}
Hence, we have
\begin{align*}
\Delta_+ = \frac{1}{\sqrt n}\frac{\sqrt{\sinh 2\xi_\tau}}{|\sinh(\xi_\tau+i\eta)|} \frac{u}{\sqrt 2}, \qquad \text{ and } \qquad
\Delta_- = \frac{1}{\sqrt n}\frac{\sqrt{\sinh 2\xi_\tau}}{|\sinh(\xi_\tau+i\eta)|} \frac{\overline v}{\sqrt 2}. 
\end{align*}
Plugging these in \eqref{eq:prop:asympIntaudX} yields
\begin{multline}  \label{eq:InErfcd=1equation1}
(1-\tau^2)^\frac{d}{2} e^{-n F(\tau)} I_{n,\tau}^d(z_\pm) 
= \frac{1}{2} \operatorname{erfc}\left(\frac{u+\overline v}{\sqrt 2}\right)\\
+ \frac{1}{6\sqrt{\pi n}} \left(\frac{\sqrt{\sinh 2\xi_\tau}}{|\sinh(\xi_\tau+i\eta)|}\right)^3 \exp\left(-\frac{(u+\overline v)^2}{2}\right)
\left(u^2-u\overline v+\overline v^2-1  
+ \mathcal O(n^{-\frac{1}{2}+3\nu})\right),
\end{multline}
as $n\to\infty$, uniformly for $z\in\partial\mathscr E_\tau$ and $u, v=\mathcal O(n^\nu)$. As shown in \cite{LeRi}, the factor in the second line of \eqref{eq:InErfcd=1equation1} can be identified with the curvature $\kappa$ of $\partial \mathscr E_\tau$ in $z$ (up to an explicit constant factor). In particular, according to \cite[C13]{LeRi} we have
\begin{align} \label{eq:expressionkappa}
\kappa = \frac{(1-\tau^2)^\frac{3}{2}}{(1+\tau^2 - 2\tau \cos 2\eta)^\frac{3}{2}}
= \left(\frac{2\tau \sinh 2\xi_\tau}{4\tau |\sinh(\xi_\tau+i\eta)|^2}\right)^\frac{3}{2}
= \frac{1}{2\sqrt 2} \left(\frac{\sqrt{\sinh 2\xi_\tau}}{|\sinh(\xi_\tau+i\eta)|}\right)^3.
\end{align}
Finally, reinstating the weight factors as in \eqref{eq:defKnInIn} and applying Lemma \ref{lem:Ftauuvz}, we obtain the result.  
\end{proof}

\begin{proof}[Proof of Theorem \ref{thm:mainThmDensity}]
As the reader may verify with some straightforward combinatorial arguments, the number of points of the DPP defined via \eqref{eq:defKerneleGinUEd} is given by the binomial coefficient $\binom{n+d-1}{d}$, which behaves as $\frac{n^d}{d!} (1+\mathcal O(1/n))$ for large $n$. For $\tau=0$, the result is more or less a direct consequence of Theorem \ref{thm:mainThmFaddeevaKernel}, where, according to \eqref{eq:proofthm:mainThmFaddeevaKernel}, we are allowed to replace $\mathcal O(1+|u|^2+|v|^2)$ by 
\begin{align*}
\frac{(u\cdot \textbf{n}+\textbf{n}\cdot v)^2 - 3 u\cdot v -1}{3} + \mathcal O(n^{-\frac{1}{2}+3\nu}).
\end{align*}
Taking $u = v = \lambda \textbf{n}$, we find the result (note that $\kappa$ reduces to the value $1$ here). 
Let us move to the case $0<\tau<1$. 
Note that the unimodular factors from Lemma \ref{lem:Ftauuvz} cancel each other when $u=v$. We can rewrite \eqref{eq:expressionkappa} as
\begin{align*}
\kappa = \left(\frac{\sqrt{\sinh 2\xi_\tau}}{|z_+^2-2|^\frac{1}{2}}\right)^3
= \frac{1}{|(|\Re z|+i|\Im z|)^2 - \frac{4\tau}{1-\tau^2}|^\frac{3}{2}}
= \frac{1}{\left((|\Re z|^2-|\Im z|^2 - \frac{4\tau}{1-\tau^2})^2+ 4 |\Re z|^2|\Im z|^2\right)^\frac{3}{4}}.
\end{align*}
The statement is now a direct consequence of Proposition \ref{prop:asympIntaudX}, and (the second part of) Lemma \ref{lem:asympPhiuvz}. 
\end{proof}

\end{document}